\documentclass[preprint,11pt,3p,authoryear]{elsarticle}

\usepackage{latexsym}
\usepackage{amsmath, amsfonts, amsthm,amssymb}
\usepackage{epsfig}
\usepackage{stmaryrd}
\usepackage{multicol}
\usepackage{pdfsync}
\usepackage{bm}
\usepackage{dsfont}
\usepackage{comment}
\usepackage{algorithmic}
\usepackage{algorithm}
\usepackage{caption}
\usepackage{psfrag}
\usepackage{subfig}
\usepackage{xcolor}
\usepackage{graphics}
\usepackage{enumitem}
\usepackage{mathtools}
\usepackage{bbm}
\usepackage{mathrsfs} 
\usepackage{color}
\usepackage{hyperref}
\usepackage[title]{appendix}
\usepackage{cancel}
\usepackage{float}
\usepackage{todonotes}
\usepackage{setspace}
\usepackage{wrapfig}

\usepackage{tikz-cd}
\journal{TBA}
\hypersetup{
	colorlinks = true, % Colours links instead boxes
	urlcolor = blue, % Colour for external hyperlinks
	linkcolor = red, % Colour of internal links
	citecolor = blue % Colour of citations
}

\newtheorem{theorem}{Theorem}[section]

\newtheorem{lemma}[theorem]{Lemma}
\newtheorem{proposition}[theorem]{Proposition}
\theoremstyle{definition}
\newtheorem{definition}[theorem]{Definition}

\theoremstyle{remark}

\numberwithin{equation}{section}
\newcommand{\diff}{\mathrm{d}}

\DeclareMathOperator*{\argmax}{arg\,max}

\title{\textbf{
Nash Equilibrium between Brokers and Traders
}\\
Forthcoming in Finance and Stochastics
}

\author[label1,label2]{\'{A}lvaro Cartea}

\address[label1]{Mathematical Institute, University of Oxford}
\address[label2]{Oxford-Man Institute of Quantitative Finance}
\ead{alvaro.cartea@maths.ox.ac.uk}

\address[label4]{Department of Statistical Sciences, University of Toronto}

\author[label4,label2]{Sebastian Jaimungal}
\ead{sebastian.jaimungal@utoronto.ca}

\author[label1,label2]{Leandro S\'{a}nchez-Betancourt}
\ead{leandro.sanchez-betancourt@kcl.ac.uk}

\begin{document}

\newcommand{\eps}{\varepsilon}
\newcommand{\la}{\left \langle}
\newcommand{\ra}{\right\rangle}
\newcommand{\cb}[1]{{\color{blue} #1}}
\newcommand{\norm}[1]{\left\lVert #1 \right\rVert}
\newcommand{\bae}{\begin{equation}\begin{aligned}}
\newcommand{\eae}{\end{aligned}\end{equation}}
\newcommand{\beq}{\begin{equation}}
\newcommand{\eeq}{\end{equation}}
\newcommand{\N}{\mathbb{N}}
\newcommand{\R}{\mathbb{R}}
\newcommand{\E}{\mathbb{E}}
\newcommand{\Pb}{\mathbb{P}}
\newcommand{\Lb}{\mathbb{L}}
\renewcommand{\H}{{\mathbb{H}^2}}

\newcommand{\mfT}{{\mathfrak{T}}}
\newcommand{\mfA}{{\mathfrak{A}}}
\newcommand{\mfB}{{\mathfrak{B}}}
\newcommand{\mfO}{{\mathfrak{O}}}
\newcommand{\tT}{{t\in\mfT}}
\newcommand{\mcA}{{\mathcal{A}}}
\newcommand{\mcC}{{\mathcal{C}}}
\newcommand{\mcF}{{\mathcal{F}}}
\newcommand{\mcB}{{\mathcal{B}}}
\newcommand{\mcH}{{\mathcal{H}}}

\newcommand{\transB}{\mathfrak{t}}
\newcommand{\decayB}{\mathfrak{p}}
\newcommand{\instantB}{\mathfrak{h}}

\newcommand{\tempB}{a}
\newcommand{\tempI}{b}
\newcommand{\tempU}{c}
\newcommand{\permB}{k}
\newcommand{\termpB}{\phi}
\newcommand{\termpI}{\psi}
\newcommand{\runnB}{r^B}
\newcommand{\runnI}{r^I}

\newcommand{\nuI}{\eta}
\newcommand{\nuB}{\nu}
\newcommand{\nuU}{\xi}
\newcommand{\nuIbis}{\kappa}
\newcommand{\nuBbis}{\zeta}
\newcommand{\nuIdir}{\mathfrak n}
\newcommand{\nuBdir}{\mathfrak v}

\newcommand{\nuBstar}{\nu^*}
\newcommand{\nuBstartilde}{\tilde\nu^*}
\newcommand{\nuIstar}{\eta^*}

\newcommand{\MI}[1][I]{M^{#1}}
\newcommand{\MB}[1][B]{M^{#1}}
\newcommand{\MBa}[1][B]{\tilde{N}^{#1}}
\newcommand{\MBb}[1][B]{\tilde{M}^{#1}}
\newcommand{\MZ}[1][Z]{M^{#1}}

\newcommand{\QI}[1][I]{Q^{#1}}
\newcommand{\QB}[1][B]{Q^{#1}}
\newcommand{\QU}[1][U]{Q^{#1}}
\newcommand{\XI}[1][I]{X^{#1}}
\newcommand{\XB}[1][B]{X^{#1}}
\newcommand{\XU}[1][U]{X^{#1}}

\newcommand{\mcFB}{\mcF}%^B
\newcommand{\mcFI}{\mcF}%^I

\newcommand{\vfB}{J^{B*}}
\newcommand{\vfI}{J^{I*}}
\newcommand{\pcB}{J^B}
\newcommand{\pcI}{J^I}
\newcommand{\mcAB}{\H}
\newcommand{\mcAI}{\H}

\newcommand{\g}{g}
\newcommand{\gI}{g^{I}}
\newcommand{\gB}{g^{B}}
\newcommand{\gZ}{g^{Z}}
\newcommand{\gY}{g^{Y}}
\newcommand{\h}{h}
\newcommand{\hI}{h^{I}}
\newcommand{\hB}{h^{B}}
\newcommand{\hZ}{h^{Z}}
\newcommand{\hY}{h^{Y}}
\newcommand{\f}{f}
\newcommand{\fI}{f^{I}}
\newcommand{\fB}{f^{B}}
\newcommand{\fZ}{f^{Z}}
\newcommand{\fY}{f^{Y}}

\renewcommand{\S}{\mathbb{S}}
\newcommand{\spaceSSS}{\mathbb{S}^{2,3}}
\newcommand{\spaceHHH}{\mathbb{H}^{2,3}}

\newcommand{\spaceK}{\mathbb{K}}
\newcommand{\spaceKsub}{\mathbb{J}}

\newcommand{\fbsdeX}{\mathcal{X}}
\newcommand{\fbsdeY}{\mathcal{Y}}
\newcommand{\fbsdeM}{\mathcal{M}}
\newcommand{\fbsdeA}{A}
\newcommand{\fbsdeAhat}{\hat{A}}
\newcommand{\fbsdeB}{B}
\newcommand{\fbsdeBhat}{\hat{B}}
\newcommand{\fbsdeb}{b}
\newcommand{\fbsdebhat}{\hat{b}}
\newcommand{\fbsdesigma}{\sigma}
\newcommand{\fbsdeG}{G}

\newcommand{\riccatiP}{\mathcal{P}}

\newcommand{\seb}[1]{{\todo[fancyline,backgroundcolor=red!20!white]{\begin{spacing}{0.5}{\tiny #1}\end{spacing}}}}
\newcommand{\leandro}[1]{{\todo[fancyline,backgroundcolor=blue!20!white]{\begin{spacing}{0.5}{\tiny #1}\end{spacing}}}}

\newcommand\leo{ \color{black}}

\newcommand\new{ \color{black}}

\DeclarePairedDelimiter\abs{\lvert}{\rvert}%

\begin{abstract}
We study the perfect information Nash equilibrium between a broker and her clients --- an informed trader and an uniformed trader. In our model, the broker trades in the lit exchange where trades have instantaneous and transient price impact with exponential resilience, while both clients  trade with the broker. The informed trader and the broker maximise expected wealth subject to inventory penalties, while the uninformed trader is not strategic and sends the broker random buy and sell orders. We characterise the Nash equilibrium of the trading strategies with the solution to a coupled system of forward-backward stochastic differential equations (FBSDEs). We solve this system explicitly and  study the effect of information, profitability, and inventory control in the  trading strategies of the broker and the informed trader.
\end{abstract}

\maketitle

\section{Introduction}

This paper characterises the Nash equilibrium  between a broker and her clients in an over-the-counter (OTC) market. The problem formulation is as follows.  The broker provides liquidity to both an uninformed trader and an informed trader, and the broker also trades in a lit exchange where her buy and sell market orders have transient and instantaneous price impact. The trading flow of the uniformed trader is an exogenous process that mean-reverts around zero. The informed trader knows the stochastic drift of the asset and trades strategically. On the other hand, the broker knows the identity of her clients and uses this information to devise an internalisation and externalisation strategy of her inventory.

The setup of the interaction between the traders and the broker is similar to that in \cite{cartea2022broker}. Specifically, the broker and the informed trader maximise expected wealth from their trading activities, while penalising inventory holdings. For the broker, this penalty protects her strategy from inventory risk, in particular toxic inventory. Similarly, for the informed trader, the inventory penalty controls how much inventory risk he is willing to bear throughout the trading horizon. In  \cite{cartea2022broker}, the broker's trades in the lit market have a permanent impact on prices, whereas here the broker's trades have both instantaneous and transient impact (see e.g., \cite{neuman2022optimal}), where the latter is modelled with exponential kernels as in \cite{obizhaeva2013optimal}.

In our model, the broker controls her speed of trading  $\nuB$ in the lit exchange and the informed trader controls the speed  $\nuI$  at which he trades with the broker. The performance criteria of the broker  and of the informed trader are $\pcB(\nuB,\nuI)$ and  $\pcI(\nuB,\nuI)$, respectively. We show that for a large set of admissible controls:  (i) for fixed trading speed $\nuB$ of the broker, the map $\nuI\to\pcI(\nuB,\nuI)$ is strictly concave, and (ii) for fixed trading speed $\nuI$ of the informed trader, the map  $\nuB\to\pcB(\nuB,\nuI)$ is strictly concave. We use this result and the G\^ateaux derivative of the functionals $\pcI,\pcB$ to characterise the best responses of the broker and the informed trader.  Then, we show that the Nash equilibrium of the strategies of the broker and of the informed trader is given by the solution to a linear forward-backward stochastic differential equation (FBSDE). We  show that there is a unique solution to the FBSDE that characterises the Nash equilibrium if and only if there is a solution to a matrix Riccati differential equation. We show existence and uniqueness of the matrix Riccati differential equation under some restrictions of model parameters. 
We also present a proof of existence and uniqueness at the level of the FBSDE with an explicit bound on the time horizon.
Finally, we use the closed-form Nash equilibrium strategies to perform simulations and showcase the performance of the trading strategies.
We compare the performance of the strategies in the Nash equilibrium with the strategies in the two-stage optimisation of \cite{cartea2022broker} and derive insights.

There are a number of studies that focus on the equilibrium between takers and makers of liquidity in financial markets. Earlier works include those of \cite{grossman1980impossibility}, \cite{kyle1985continuous},  \cite{kyle1989informed}, and \cite{o1998market}.  \cite{grossman1988liquidity} study the liquidity of a traded asset as a result of the demand and supply of immediacy in a market where liquidity providers charge a premium for intermediating trades between two consecutive trading periods. Recently, \cite{bank2021liquidity} extend the Grossman s-Miller model to a continuous-time framework where a risky asset  is traded in two markets, a lit exchange,  and an OTC dealer market.   Similar to Bank et al., our  work considers a broker  who decides how to internalise and externalise trades. While in their model, ``pure internalisers only trade in the dealer market, whereas externalisation corresponds to offsetting trades in the open market'', see also \cite{butz2019internalisation}, in our model, the broker's strategy consists of internalising and externalising trades, and of speculative trades that are based on the information learned from the flow of the informed trader. {\new{See \cite{cartea2022broker} for details on the trading mechanisms, and for how brokers extract information about the informed trader's signal. To classify clients into informed or uninformed traders, brokers may carry out an analysis like that of Section 2 in \cite{cartea2023detecting}. A practical example of the setup in which we work is as follows: the broker may be providing liquidity in spot FX to their clients (e.g., LMAX broker in the EUR/USD pair)  and simultaneously trading in a venue with a central limit order book  (e.g., LMAX exchange in the limit order book for EUR/USD), furthermore, the broker may profile her clients into two groups of traders according to the profitability of their trades. An article with a similar formulation is \cite{barzykin2023algorithmic}, where the authors use trading speeds to externalise inventory in the lit market and they employ quotes à la Avellaneda-Stoikov to provide liquidity to  clients.}}

\cite{herdegen2023liquidity} study a game between dealers that compete for the order flow of a client.
Other stochastic games in the algorithmic trading literature study the interactions between agents liquidating inventory positions; see e.g., \cite{cont2023fast,jaber2023equilibrium}. The interaction between the broker and traders in this paper departs from this branch of the literature. Here, the incentives to trade are as follows: the informed trader profits from the private signal about the trend in the value of the asset, the uninformed trader is not strategic, and the broker manages inventory and  uses the trading speed of the informed trader to send speculative trades to the lit market.\footnote{This is in contrast to the work of \cite{cartea2023detecting} where machine learning techniques determine if individual trades are informed or uninformed.  }

Another branch in the literature  focuses on how a trader unwinds stochastic order flow.  \cite{cartea2020trading} derive an optimal liquidation strategy for a broker who makes liquidity and trades in a triplet of currency pairs, while managing stochastic order flow from her clients, in the presence of model uncertainty. In \cite{cartea2022double} the stochastic order flow is the proceeds from trading a stock in a foreign currency, so the broker must also manage exchange rate risk. 
\cite{nutz2023unwinding}  solve a control problem for the optimal externalisation schedule of an exogenous order flow with an Obizhaeva--Wang-type price impact and quadratic instantaneous costs. Recently, \cite{barzykin2024unwinding} employ stochastic filtering to solve the control problem of unwinding stochastic order flow with unobserved toxicity.
In our paper, 
the broker's objective is to manage the stochastic order flow she receives from her clients; the optimal strategies are those that balance hedging, speculative trading, and inventory control; similar to the previous two works, we also employ Obizhaeva–Wang-type price impact with quadratic instantaneous costs.

Finally, in a mean field setup,  \cite{baldacci2023mean} study a  model with one market maker and the average behaviour of an infinite number of liquidity takers. In a similar vein, 
\cite{bergault2024mean} characterise the mean-field Nash equilibrium of the strategies between  a broker, whose trades have instantaneous and  permanent price impact, and a large number of informed traders. In this paper we study the optimal responses of a broker and and informed trader,  which represents a collection of informed traders, where the broker's trades have instantaneous and transient price impact with exponential resilience.

The remainder of the paper proceeds as follows. Section \ref{sec: the model} introduces the model, the performance criteria, and defines the notion of a Nash equilibrium in our setup. 
Section \ref{sec: Nash} characterises the Nash equilibrium of the game and solves it explicitly. First, we show that both functionals are strictly concave, then, we compute their G\^ateaux derivatives, we characterise the Nash equilibrium with a system of FBSDEs, and we study existence and uniqueness of the system. 
Section \ref{sec: closed form Nash vector not} derives the closed-form solution and connects the existence and uniqueness of the FBSDE with the existence and uniqueness of a matrix Riccati differential equation; we show existence and uniqueness of this differential equation under some restrictions of model parameters. Finally, Section \ref{sec: numerics} shows numerical experiments and concludes.

\section{The model}\label{sec: the model}

Let $T>0$ be a given trading horizon and let $\mfT = [0,T]$. We fix a filtered probability space $\left(\Omega, \allowbreak\mcF, (\mcF_t)_{t\in\mfT}, \Pb \right)$ satisfying the usual conditions of right-continuity and completeness. The filtration $(\mcF_t)_{t\in\mfT}$ is the sigma-algebra generated by all processes below.

The broker trades with speed $(\nuB_t)_{t\in\mfT}$, the informed trader trades with speed $(\nuI_t)_{t\in\mfT}$, and the uninformed trader trades with speed $(\nuU_t)_{t\in\mfT}$. We assume that $\nuB,\nuI,\nuU\in \H$ where
\begin{equation}
    \H=\Big\{(\gamma_t)_{t\in\mfT} \,\,|\,\, \gamma\text{ is } \mcFB-\text{progressively measurable, and } \|\gamma\|^2:=\E\left[\textstyle\int_0^T\left(\gamma_s\right)^2\diff  [M^S,M^S]_s\right]\,<\,+\infty\Big\}\,,
\end{equation}
where  $(M^S_t)_{t\in\mfT}$ is a square-integrable martingale with respect to $(\mcF_t)_{t\in\mfT}$.
The midprice process is denoted by $(S_t)_{t\in\mfT}$ and is given by
\begin{eqnarray} \label{eqn: midprice dynamics}
S_t = S_0 + \underbrace{ \int_0^t \alpha_s\,\diff s }_{\text{signal}} + \underbrace{ Y_t }_{\text{impact}} + M^{S}_t\,,
\end{eqnarray}
where the broker's price price impact $(Y_t)_{t\in\mfT}$ satisfies
\begin{equation}\label{eq: transient}
\diff Y_t = ( \instantB\,\nuB_t-\decayB\,Y_t )\,\diff t  \,,\quad Y_0\in\mathbb{R}\,,
\end{equation}
and the trading signal process $(\alpha_t)_{t\in\mfT}$ is  square-integrable and progressively measurable.\footnote{{\new{We made the simplifying assumption that the signal enters the price as $\int_0^t \alpha_s\,\diff s$ as opposed to having it as a general stochastic process. This is a common simplifying assumption; see e.g., \cite{cartea2016incorporating}. }}} 
Here, $\decayB\geq 0$ is a  decay coefficient, and $\instantB\geq 0$ is the {\new{parameter controlling the instantaneous effect of the broker's trading on the transient impact}}. When $\decayB=0$ the price impact is permanent and linear in the integral of the broker's speed of trading.

The inventory process $\QB$ and cash process $\XB$ of the broker satisfy the equations
\begin{align}
    \diff \QB_t &= \left(\nuB_t - \nuI_t - \nuU_t\right)\diff t\,,
    & Q^B_0 = q^B\in\mathbb{R}\,,\label{eq: inventory of broker}\\
    \diff \XB_t &= - \left(S_t + \tempB\,\nuB_t\right)\nuB_t\,\diff t +   \left(S_t + \tempI\,\nuI_t\right)\nuI_t\,\diff t +   \left(S_t + \tempU\,\nuU_t\right)\nuU_t\,\diff t\,,
    & X^B_0 = 0\,,\label{eq: cash of broker}
\end{align}
where the positive constants $\tempB,\tempI,\tempU$ represent instantaneous transaction costs. We assume that the trading speed of the uninformed trader is square-integrable. Similarly, the inventory process $\QI$ and cash process $\XI$ of the informed trader satisfy the equations
\begin{align}
    \diff \QI_t &= \nuI_t\,\diff t\,, 
    & Q^I_0 = q^I\in\mathbb{R}\,,\label{eq: inventory of informed}\\
    \diff \XI_t &= - \left(S_t + \tempI\,\nuI_t\right)\nuI_t\,\diff t\,,
    & X^I_0 = 0\,.\label{eq: cash of informed}
\end{align}

Sometimes we write the control in the superscript when we highlight the controls that affect a process, for example,  we write $\XI[I,\nuI]$ and $\QI[I,\nuI]$ for the cash process and inventory process of the informed trader.
The sets of admissible strategies for the informed trader and the broker are both given by the set $\H$. 

Let $\nuB\in\mcAB$ and let $\nuI\in\mcAI$. The performance criterion of the informed trader is  $\pcI(\nuB,\nuI)$, given by
\begin{equation}\label{eq: performance criterion informed}
    \pcI(\nuB,\nuI) = \E\left[
    \XI_T + \QI_T\,S_T - \termpI\,\left(\QI_T\right)^2 - \runnI \int_0^T\left(\QI_s\right)^2\,\diff s 
    \right]\,,
\end{equation}
for $\termpI,\runnI$ non-negative inventory control constants.
Similarly, the performance criterion of the broker is  $\pcB(\nuB,\nuI)$, given by
\begin{equation}\label{eq: performance criterion broker}
    \pcB(\nuB,\nuI) = \E\left[
    \XB_T + \QB_T\,S_T - \termpB\,\left(\QB_T\right)^2 - \runnB \int_0^T\left(\QB_s\right)^2\,\diff s 
    \right]\,,
\end{equation}
for $\termpB,\runnB$ non-negative inventory control constants. 
We let 
\begin{equation}
    \varphi = \termpB - \frac{1}{2}\,\instantB\,,
\end{equation}
and we assume that $\varphi \geq 0$ and that $\tempB>\decayB\,\instantB\,T^2$.
These are non-restrictive technical assumptions  which are sufficient conditions to guarantee strict concavity of $J^B$ up to null sets, and that we also use to prove existence of a non-symmetric matrix Riccati differential equation that arises in the closed-form solution to Nash equilibrium of the problem.
In particular, the assumption that $\tempB>\decayB\,\instantB\,T^2$ is not restrictive because one expects  $\tempB \gg \instantB\,\decayB$; see, e.g., page 148 in \cite{cartea2015algorithmic}. {\new{The assumption $\varphi\geq 0$ ensures that, effectively, the broker penalises holding terminal inventory, as opposed to rewarding holding terminal inventory.  For example, when $\decayB=0$, the constraint $\varphi\geq 0$  implies that the overall terminal penalty on inventory (once we consider the effect of the permanent price impact) is non-negative; see bottom of page 148 in \cite{cartea2015algorithmic}. }
}

\begin{lemma}[Finiteness of Performance Criterion]
\label{lemma:finiteness}
    The functional  $\pcI:\mcAB\times\mcAI\to \R$ can be written as
\begin{align}\label{eq: repres pcI}
    \pcI(\nuB,\nuI) = S_0\,q^I - \termpI\left(q^I\right)^2+\E\left[
    \int_0^T \!\!\left\{-\tempI\,\nuI^2_t  + \QI_t\,\Big(\alpha_t+ ( \instantB\,\nuB_t-\decayB\,Y_t )
    -2\,\termpI\,\nuI_t -\runnI\,\QI_t 
    \Big)    \right\}\diff t 
    \right].
\end{align}
    Similarly, the functional $\pcB:\mcAB\times\mcAI\to \R$ can be written as 
\begin{align}
    \pcB(\nuB,\nuI) &= S_0\,q^B - \termpB\,\left(q^B\right)^2 \label{eq: repres pcB}\\
    &\quad + \E\left[
    \int_0^T \!\!\left\{-\tempB\,\nuB^2_t + \tempI\,\nuI^2_t + \tempU\,\nuU^2_t + \QB_t\,\Big(\alpha_t+ ( \instantB\,\nuB_t-\decayB\,Y_t )
    -2\,\termpB\,(\nuB_t - \nuI_t -\nuU_t) -\runnB\,\QB_t
    \Big)   \right\}\diff t 
    \right]\!.\nonumber
\end{align}
\end{lemma}

\begin{proof}
The representation follows immediately from the product rule and the dynamics above. To see that  for $\nuI\in\mcAI$ and $\nuB\in\mcAB$ both $\pcI(\nuB,\nuI)\in\R$ and $\pcB(\nuB,\nuI)\in\R$ (i.e., the performance criterion are finite), we proceed as follows. Use the inequality $2\,|x\,y|\leq x^2+y^2$, together with the following three results:
\begin{enumerate}[label=(\roman*)]
    \item if $\nuI\in\mcAI$, then
    \begin{align}
    \left(\QI[I,\nuI]_t\right)^2 &= \left(q_0^I\right)^2 + 2\,q_0^I\,\int_0^t\nuI_s\,\diff s + \left(\int_0^t\nuI_s\,\diff s\right)^2
    \\
    &\leq \left(q_0^I\right)^2 + \left|q_0^I\right|\,\int_0^t\left(1+\nuI^2_s\right)\,\diff s + \int_0^t\left(\nuI_s\right)^2\diff s\,,
    \end{align}
    thus, there exists $C\in\R^+$ such that 
    \begin{align}
        \left|\QI[I,\nuI]_t\right|^2 \leq C\left(1 + \int_0^t|\nuI_s|^2\,\diff s\right)\,,
    \end{align}
    which implies that 
\begin{align}
\E\left[\int_0^T\left|\QI[I,\nuI]_t\right|^2\,\diff t\right] \leq C\,T\left(1 + \E\left[\int_0^T|\nuI_s|^2\,\diff s\right]\right)<+\infty\,.
\end{align}
because $\nuI\in\H$.
    \item if $\nuI\in\mcAI$, $\nuB\in \mcAB$, and $\nuU$ is square-integrable, similar to the  above, there exists $C\in\R^+$ such that 
\begin{align}
\E\left[\int_0^T\left|\QB[B,\nuB]_t\right|^2\,\diff t\right] \leq C\,T\left(1 + \E\left[\int_0^T|\nuI_s|^2\,\diff s + \int_0^T|\nuB_s|^2\,\diff s+ \int_0^T|\nuU_s|^2\,\diff s\right]\right)<+\infty\,.
\end{align}
\item if $\nuB\in\mcAB$, then, from \eqref{eq: transient},
    \begin{align}
        Y_t = e^{-\decayB\,t}\,Y_0 + \instantB\,\int_0^t \nuB_s\,e^{-\decayB\,(t-s)}\,\diff s\,.
    \end{align}
    Then, there exists $C\in\R^+$ such that 
    \begin{align}
       \E\left[ \int_0^T |Y_t|^2\diff t\right] \leq C\,T\left(1+ \E\left[\int_0^t |\nuB_s|^2\,\diff s\right]\right)<+\infty.
    \end{align}
\end{enumerate}

Finally, using that $\nuI,\nuU,\nuB,\QI[I,\nuI],\QB[B,\nuB],\alpha, Y$ are all square-integrable, together with the representations in \eqref{eq: repres pcI} and \eqref{eq: repres pcB} concludes the proof.
\end{proof}

\begin{definition}\label{def: Nash equilibrium}
The pair of trading speeds $\left(\nuBstar, \nuIstar\right)\in \mcAB\times\mcAI$ is a \textbf{Nash equilibrium} of the strategies between the informed trader and the broker if the following two conditions hold:
\begin{enumerate}[label=(\roman*)]
    \item for any speed of trading $\nuB\in\mcAB$
    \begin{equation}
        \pcB(\nuB, \nuIstar) \leq \pcB(\nuBstar, \nuIstar)\,,
    \end{equation}
    \item for any speed of trading $\nuI\in\mcAI$
    \begin{equation}
        \pcI(\nuBstar, \nuI) \leq \pcI(\nuBstar, \nuIstar)\,.
    \end{equation}
\end{enumerate}
\end{definition}

In the next section we characterise  the perfect information  Nash equilibrium of the trading strategies.

\section{Nash equilibrium}\label{sec: Nash}

\subsection{Characterisation of equilibrium}

We employ G\^ateaux derivatives to characterise the Nash equilibrium between the broker and the informed trader.  Let $\nuB,\nuB'\in\mcAB$ and $\nuI,\nuI'\in\mcAI$. The directional derivative of $\pcB$ at $\nuB$ in the direction of $\nuBdir:=(\nuB'-\nuB)$ is given by
\begin{equation}
\langle \mathcal{D}\,\pcB(\nuB,\nuI),\nuBdir\rangle=\lim_{\eps\to 0}\tfrac{1}{\eps}\left[\pcB(\nuB+ \eps\,\nuBdir, \nuI)-\pcB(\nuB,\nuI)\right]\,,
\end{equation}
when the limit exists. Similarly, the directional derivative of $\pcI$ at $\nuI$ in the direction of $\nuIdir:=(\nuI'-\nuI)$ is given by
\begin{equation}
\langle \mathcal{D}\,\pcI(\nuB,\nuI),\nuIdir\rangle=\lim_{\eps\to 0}\tfrac{1}{\eps}\left[\pcI(\nuB, \nuI+ \eps\,\nuIdir)-\pcI(\nuB,\nuI)\right]\,,
\end{equation}
when the limit exists.

First, we show that for fixed $\nuB\in\mcAB$ the functional $\nuI\to\pcI(\nuB,\nuI)$ is strictly concave; this is proven in the next proposition.

\begin{proposition}\label{prop: pcI concave}
Let $\nuB\in\mcAB$, the functional $\pcI(\nuB,\,\,\cdot\,):\mcAI\to\R$ is strictly concave.
\end{proposition}
\begin{proof}
Let $\nuB\in\mcAB$ and $\nuI,\nuIbis \in\mcAI$. Let $A\in \mcFI\otimes \mcB(\mfT)$  with $\mu(A) > 0$, where $\mu:=\Pb\otimes\diff t$ and for all $(\omega,t)\in A$  we have that $\nuI_t(\omega) \neq \nuIbis_t(\omega)$, i.e., $\nuI$ and $\nuIbis$ differ on a set with non-zero $\mu$-measure. 
Let $\rho \in (0,1)$, we need to show that 
\begin{equation}
    \pcI(\nuB, \rho\,\nuI + (1-\rho)\,\nuIbis) > \rho\,\pcI(\nuB, \nuI) + (1-\rho)\,\pcI(\nuB,\nuIbis)\,.
\end{equation}
First, from linearity of \eqref{eq: inventory of informed}, observe that 
\begin{equation}
    \QI[I,\rho\,\nuI + (1-\rho)\,\nuIbis]_t = \rho\,\QI[I,\nuI]_t + (1-\rho)\,\QI[I,\nuIbis]_t\,,
\end{equation}
thus,
\allowdisplaybreaks
\begin{align}
    &\pcI(\nuB, \rho\,\nuI + (1-\rho)\,\nuIbis) \\
    &\quad =S_0\,q^I - \termpI\left(q^I\right)^2+ \E\Bigg[\int_0^T \bigg\{-\tempI\,\left(\rho\,\nuI_t + (1-\rho)\,\nuIbis_t\right)^2  + \left(\rho\,\QI[I,\nuI]_t + (1-\rho)\,\QI[I,\nuIbis]_t\right)\,(\alpha_t + ( \instantB\,\nuB_t-\decayB\,Y_t )) 
    \nonumber\\
    &\hspace*{14em}
    -2\,\termpI\,\left(\rho\,\QI[I,\nuI]_t + (1-\rho)\,\QI[I,\nuIbis]_t\right)\,\left(\rho\,\nuI_t + (1-\rho)\,\nuIbis_t\right) 
    \nonumber
    \\
    &\hspace*{14em} -\runnI\,\left(\rho\,\QI[I,\nuI]_t + (1-\rho)\,\QI[I,\nuIbis]_t\right)^2  \bigg\}\,\diff t  
    \Bigg]
    \\
    &\quad = \rho\,\pcI(\nuB, \nuI ) + \rho\,(1-\rho)\,\E\left[\int_0^T \left\{ \tempI\,\nuI^2_t + 2\,\termpI\,\QI[I,\nuI]_t\,\nuI_t + \runnI\,\left(\QI[I,\nuI]_t\right)^2\right\}\diff t 
    \right]
    \nonumber
    \\
    &\hspace*{6.25em}
    + (1-\rho) \,\pcI(\nuB, \nuIbis ) + \rho\,(1-\rho)\E\left[\int_0^T \left\{ \tempI\,\nuIbis^2_t + 2\,\termpI\,\QI[I,\nuIbis]_t\,\nuIbis_t + \runnI\,\left(\QI[I,\nuIbis]_t\right)^2\right\}\diff t
    \right] 
    \nonumber\\
    &\hspace*{6.25em}  - 2\, \rho\,(1-\rho)\,\E\left[\int_0^T \left\{ \tempI\,\nuI_t\,\nuIbis_t + \termpI\,\QI[I,\nuI]_t\,\nuIbis_t +  \termpI\,\QI[I,\nuIbis]_t\,\nuI_t + \runnI\,\QI[I,\nuIbis]_t\,\QI[I,\nuI]_t\right\}\diff t
    \right] \\
    &\quad = \rho\,\pcI(\nuB, \nuI ) + (1-\rho)\,\pcI(\nuB, \nuIbis ) \nonumber\\
    &\qquad \qquad + \rho\,(1-\rho)\,\E\left[\int_0^T\left\{ \tempI\left(\nuI_t - \nuIbis_t \right)^2 + 2\,\termpI\left( \QI[\nuI]_t - \QI[\nuIbis]_t\right) \left( \nuI_t - \nuIbis_t\right) +  \runnI\,\left(\QI[I,\nuI]_t- \QI[I,\nuIbis]_t\right)^2   \right\}\diff t
    \right]\,.
\end{align}
Then, by the definition of $\QI$ we have 
\begin{align}
    \E\left[\int_0^T \left( \QI[I,\nuI]_t - \QI[I,\nuIbis]_t\right) \left( \nuI_t - \nuIbis_t\right) \diff t
    \right] & = \E\left[\int_0^T \left( \int_0^t (\nuI_u - \nuIbis_u)\diff u\right) \left( \nuI_t - \nuIbis_t\right) \diff t
    \right] 
    \nonumber \\
    & = \frac{1}{2}\,\E\left[\int_0^T \int_0^T \left(\nuI_u - \nuIbis_u\right) \left( \nuI_t - \nuIbis_t\right) \diff u\, \diff t
    \right] 
    \nonumber \\
    & = \frac{1}{2}\,\E\left[\left(\int_0^T (\nuI_t - \nuIbis_t)\, \diff t\right)^2
    \right]\,, \label{eqn:int-q-eta}
\end{align}
where the equality in the second line follows because $\nuI,\nuIbis \in \mcAI$, hence, we apply Fubini's theorem.
Therefore, 
\begin{align*}
   \pcI(\nuB, \rho\,\nuI + (1-\rho)\,\nuIbis) =& \;\;
   \rho\,\pcI(\nuB, \nuI ) + (1-\rho)\,\pcI(\nuB, \nuIbis ) 
   \\
    &+ \rho\,(1-\rho)\,\E\left[\int_0^T\left\{ \tempI\left(\nuI_t - \nuIbis_t \right)^2 +  \runnI\,\left(\QI[I,\nuI]_t- \QI[I,\nuIbis]_t\right)^2   \right\}\diff t
    \right] 
    \\ 
    & + \rho\,(1-\rho)\,\termpI\,\E\left[\left(\int_0^T \nuI_t - \nuIbis_t\, \diff t\right)^2
       \right]\,.
\end{align*}
Hence, $\pcI(\nuB, \rho\,\nuI + (1-\rho)\,\nuIbis) > \rho\,\pcI(\nuB, \nuI) + (1-\rho)\,\pcI(\nuB,\nuIbis)$ because $\mu(A) > 0$ and $\E\left[\int_0^T \left(\nuI_t - \nuIbis_t\right)^2 \diff t\right] >0$.
\end{proof}

The proposition below shows that for fixed $\nuI\in\mcAI$ the broker's functional $\nuB\to\pcB(\nuI,\nuB)$ is also strictly concave.
\begin{proposition}\label{prop: pcB concave}
Let $\nuI\in\mcAI$, the functional $\pcB(\,\cdot\,,\, \nuI):\mcAB\to\R$ is strictly concave.
\end{proposition}
\begin{proof}
    Let $\nuI\in\mcAI$ and $\nuB,\nuBbis \in\mcAB$. Let $A\in \mcFB\otimes \mcB(\mfT)$  with $\mu(A) > 0$, where $\mu:=\Pb\otimes\diff t$, and  for all $(\omega,t)\in A$ we have that $\nuB_t(\omega) \neq \nuBbis_t(\omega)$, i.e., $\nuB$ and $\nuBbis$ differ on a set with non-zero $\mu$-measure.
Let $\rho \in (0,1)$, we need to show that 
\begin{equation}
    \pcB(\rho\,\nuB + (1-\rho)\,\nuBbis, \nuI) > \rho\,\pcB(\nuB, \nuI) + (1-\rho)\,\pcB(\nuBbis,\nuI)\,.
\end{equation}
First, from linearity of \eqref{eq: inventory of broker}, observe that 
\begin{equation}
    \QB[B,\rho\,\nuB + (1-\rho)\,\nuBbis]_t = \rho\,\QB[B,\nuB]_t + (1-\rho)\,\QB[B,\nuBbis]_t
\end{equation}
and, from linearity of \eqref{eq: transient}, 
\begin{equation}
   Y^{\rho\,\nuB + (1-\rho)\,\nuBbis}_t = e^{-\decayB\,t}\,Y_0 + \instantB\,\int_0^t e^{-\decayB\,(t-s)}\,\left(\rho\,\nuB_s + (1-\rho)\,\nuBbis_s\right)\,\diff s  = \rho\,Y^\nuB_t + (1-\rho)\,Y^\nuBbis_t\,.
\end{equation}
Next, 
\allowdisplaybreaks
\begin{align}
&\pcB(\rho\,\nuB + (1-\rho)\,\nuBbis, \nuI) = S_0\,q^B - \termpB\,\left(q^B\right)^2 \\
&\; +\E\Bigg[\int_0^T \bigg\{  -\tempB\,\left(\rho\,\nuB_t + (1-\rho)\,\nuBbis_t\right)^2 + \tempI\,\nuI^2_t + \tempU\,\nuU^2_t  + \left(\rho\,\QB[B,\nuB]_t + (1-\rho)\,\QB[B,\nuBbis]_t\right)\,\alpha_t \nonumber \\
&\quad\qquad + \left(\rho\,\QB[B,\nuB]_t + (1-\rho)\,\QB[B,\nuBbis]_t\right)\,\left(\instantB\,\rho\,\nuB_t + \instantB\,(1-\rho)\,\nuBbis_t-\decayB\,\rho\,Y^\nuB_t-\decayB\,(1-\rho)\,Y^\nuBbis_t \right)\\
&\quad\qquad  -2\,\termpB\,\left(\rho\,\QB[B,\nuB]_t + (1-\rho)\,\QB[B,\nuBbis]_t\right)\,\Big(\rho\,\nuB_t + (1-\rho)\,\nuBbis_t - \nuI_t -\nuU_t\Big)  -\runnB\,\left(\rho\,\QB[B,\nuB]_t + (1-\rho)\,\QB[B,\nuBbis]_t\right)^2   \bigg\}\,\diff t
\Bigg] 
\nonumber \\
& = \rho\,\pcB(\nuB, \nuI) + \rho\,(1-\rho)\,\E\left[\int_0^T \left\{ \tempB\,\nuB^2_t + \decayB\,\QB[B,\nuB]_t\,Y^{\nuB}_t - \instantB\,\QB[B,\nuB]_t\,\nuB_t + 2\,\termpB\,\QB[B,\nuB]_t\,\nuB_t + \runnB\,\left(\QB[B,\nuB]_t\right)^2 \right\} \diff t
\right]
\nonumber \\
&\quad + (1-\rho)\,\pcB(\nuBbis, \nuI) + \rho\,(1-\rho)\,\E\left[\int_0^T \left\{ \tempB\,\nuBbis^2_t - \instantB\,\QB[B,\nuBbis]_t\,\nuBbis_t + \decayB\,\QB[B,\nuBbis]_t\,Y^{\nuBbis}_t + 2\,\termpB\,\QB[B,\nuBbis]_t\,\nuBbis_t   + \runnB\,\left(\QB[B,\nuBbis]_t\right)^2 \right\} \diff t
\right] 
\nonumber\\
&\quad - \rho\,(1-\rho)\,\E\Bigg[\int_0^T \bigg\{ 2\,\tempB\,\nuB_t\,\nuBbis_t - \instantB\,\left(\QB[B,\nuBbis]_t\,\nuB_t + \QB[B,\nuB]_t\,\nuBbis_t\right) + \decayB\,\left(\QB[B,\nuBbis]_t\,Y^{\nuB}_t + \QB[B,\nuB]_t\,Y^{\nuBbis}_t\right) \\
&\qquad\hspace{14em} + 2\,\termpB \,\left(\QB[B,\nuB]_t\,\nuBbis_t  + \QB[B,\nuBbis]_t\,\nuB_t\right)  + 2\,\runnB\,\QB[B,\nuB]_t\,\QB[B,\nuBbis]_t \bigg\}\,\diff t
\Bigg] 
\nonumber\\
&= \rho\,\pcB(\nuB, \nuI) + (1-\rho)\,\pcB(\nuBbis, \nuI) 
\nonumber\\
&\qquad
+ \rho\,(1-\rho)\, \E\Bigg[ \int_0^T \bigg\{ \tempB\left(\nuB_t - \nuBbis_t\right)^2 +(2\,\termpB- \instantB)\,\left(\QB[B,\nuB]_t - \QB[B,\nuBbis]_t \right)\left(\nuB_t -\nuBbis_t\right) 
\\
&\hspace{11em}+ \decayB\,\left(\QB[B,\nuB]_t - \QB[B,\nuBbis]_t \right)\left(Y^\nuB_t -Y^\nuBbis_t\right)  +  \runnB \left(\QB[B,\nuB]_t - \QB[B,\nuBbis]_t\right)^2  \bigg\}\,\diff t
\label{eqn:JB-inequality}
\Bigg]\,,
\end{align}
because $\mu(A)>0$ and $\E\left[ \int_0^T \left(\nuB_t - \nuBbis_t\right)^2\diff t
\right]>0$. Further, as in \eqref{eqn:int-q-eta},
\[
\E\left[\textstyle\int_0^T \left(\QB[B,\nuB]_t - \QB[B,\nuBbis]_t \right)\left(\nuB_t -\nuBbis_t\right) \diff t
\right] \allowbreak = \tfrac{1}{2}\,\E\left[\left( \textstyle\int_0^T (\nuB_t -\nuBbis_t )\,\diff t\right)^2
\right]\geq 0.
\]
Moreover,
\begin{align}
\Bigg|\E\left[\textstyle\int_0^T \left(\QB[B,\nuB]_t - \QB[B,\nuBbis]_t \right)\left(Y^\nuB_t -Y^\nuBbis_t\right) \diff t
\right] \Bigg| &\leq \E\left[\textstyle\int_0^T \left(\textstyle\int_0^t |\nuB_s-\nuBbis_s| \,\diff s \right)\left(\instantB\,\textstyle\int_0^t e^{-\decayB\,(t-u)}|\nuB_u-\nuBbis_u |\diff u\right) \diff t
\right] \nonumber \\
&\leq T^2\,\instantB\,\E\left[\textstyle\int_0^T \left(\nuB_s-\nuBbis_s\right)^2 \,\diff s 
\right]\,.
\end{align}

With these inequalities and \eqref{eqn:JB-inequality}, it follows that 
\begin{equation}
  \pcB(\rho\,\nuB + (1-\rho)\,\nuBbis, \nuI) >  \rho\,\pcB(\nuB, \nuI) + (1-\rho)\,\pcB(\nuBbis, \nuI)\,,
\end{equation}
because $2\,\termpB -\instantB\geq 0$ and $\tempB>\decayB\,\instantB\,T^2$, which concludes the proof.
\end{proof}

Armed with strict the concavity of both $\pcI$ and $\pcB$ (in the relevant variables), we study the G\^ateaux derivatives of both functionals.

\begin{proposition}[G\^ateaux derivative of informed trader's functional]\label{prop: Gateaux pcI}
Let $\nuB\in\mcAB$ and $\nuI,\nuI'\in\mcAI$. The G\^ateaux derivative of $\pcI$ at $\nuI$ in the direction of $\nuIdir:=(\nuI'-\nuI)$ is given by
\begin{align}
    \langle \mathcal{D}\pcI(\nuB,\nuI),\nuIdir\rangle =  \E\left[\int_0^T \nuIdir_t \left\{-2\,\tempI\,\nuI_t - 2\,\termpI\,\QI_t + 
    \int_t^T \left(\alpha_u +  \instantB\,\nuB_t-\decayB\,Y_u   - 2\,\termpI\,\nuI_u - 2\,\runnI\,\QI_u\right) \diff u 
    \right\} \diff t
    \right]\,.
\end{align}
\end{proposition}

\begin{proof} 
    Let $\nuI,\nuI' \in \mcAI$ and $\nuIdir:=(\nuI'-\nuI)$, let $\nuB\in\mcAB$, and let $\eps>0$. Next,  for $t\in\mfT$ we have
\begin{equation}
\QI[I,\nuI+\eps\,\nuIdir]_t = \QI[I,\nuI]_t + \eps\,\int_0^t \nuIdir_u \,\diff u
\end{equation}
and \allowdisplaybreaks
\begin{align*}
 \pcI(\nuB,\nuI+\eps\,\nuIdir) - \pcI(\nuB,\nuI) & =\E\Bigg[
    \int_0^T \bigg\{-2\,\tempI\,\eps\,\nuI_t\,\nuIdir_t  -\tempI\,\eps^2\,\nuIdir^2_t  + \eps\,\int_0^t\nuIdir_u\,\diff u\,(\alpha_t+ \instantB\,\nuB_t-\decayB\,Y_t ) \\ 
&\qquad \qquad -2\,\termpI\,\left(\QI[I,\nuI]_t\,\eps\,\nuIdir_t + \eps\,\int_0^t\nuIdir_u\,\diff u\,\,\nuI_t +\eps^2\,\nuIdir_t\,\int_0^t\nuIdir_u\,\diff u\right) \\
&\qquad \qquad\qquad  -\runnI\,\left(2\,\eps\,\QI[I,\nuI]_t\,\int_0^t\nuIdir_u\,\diff u + \eps^2\,\left(\int_0^t\nuIdir_u\,\diff u\right)^2 \right) \bigg\}\,\diff t 
\Bigg] \\
&\quad  = \eps \,\E\Bigg[\int_0^T \bigg\{
-2\,\tempI\,\nuI_t\,\nuIdir_t   + \int_0^t\nuIdir_u\,\diff u\,(\alpha_t+  \instantB\,\nuB_t-\decayB\,Y_t )\\
&\qquad \qquad -2\,\termpI\,\left(\QI[I,\nuI]_t\,\nuIdir_t + \int_0^t\nuIdir_u\,\diff u\,\nuI_t \right) -2\,\runnI\,\QI[I,\nuI]_t\,\int_0^t\nuIdir_u\,\diff u 
\bigg\}\,\diff t
\Bigg] + \mathcal{O}\left(\eps^2\right)\,.
\end{align*}
The coefficient for $\eps^2$ is finite by Lemma \ref{lemma:finiteness}.
It then follows that 
\begin{align}
    \langle \mathcal{D}\,\pcI(\nuB,\nuI),\nuIdir\rangle &= \E\Bigg[\int_0^T \bigg\{
-2\,\tempI\,\nuI_t\,\nuIdir_t   + \int_0^t\nuIdir_u\,\diff u\,\left(\alpha_t+ \instantB\,\nuB_t-\decayB\,Y_t  -2\,\termpI\,\nuI_t -2\,\runnI\,\QI[I,\nuI]_t
\right) -2\,\termpI\, \QI[I,\nuI]_t\,\nuIdir_t 
\bigg\}\,\diff t
\Bigg]\nonumber\\
& = \E\Bigg[\int_0^T \nuIdir_t\,\bigg\{ 
-2\,\tempI\,\nuI_t  -2\,\termpI\, \QI[I,\nuI]_t + 
\int_t^T\left(\alpha_u+  \instantB\,\nuB_t-\decayB\,Y_t -2\,\termpI\,\nuI_u -2\,\runnI\,\QI[I,\nuI]_u
\right)\diff u 
\bigg\}\,\diff t
\Bigg]\,,
\end{align}
where the last line is a consequence of Fubini's theorem.
\end{proof}

\begin{proposition}[G\^ateaux derivative of broker's functional]\label{prop: Gateaux pcB}
Let $\nuB,\nuB'\in\mcAB$ and $\nuI\in\mcAI$. The G\^ateaux derivative of $\pcB$ at $\nuB$ in the direction of $\nuBdir:=(\nuB'-\nuB)$ is given by
\begin{align}
    \langle \mathcal{D}\,\pcB(\nuB,\nuI),\nuBdir\rangle =  \E\Bigg[
    & \int_0^T \nuBdir_t\,\bigg\{-2\,\tempB\,\nuB_t   - (2\,\termpB-\instantB)\,\QB[B,\nuB]_t  - \decayB\,\instantB\,\int_t^T \QB[B,\nuB]_u \,e^{-\decayB\,(u-t)}\,\diff u \\
    &\quad  +  \int_t^T \left(  \alpha_u
    - (2\,\termpB-\instantB)\,\nuB_u
    -\decayB\,Y_u +2\,\termpB\,\left(\nuI_u+\nuU_u\right)
    -2\,\runnB\,\QB[B,\nuB]_u    
    \right)\diff u  
 \bigg\}\,\diff t
 \Bigg]\,.\nonumber
\end{align}
\end{proposition}

\begin{proof}
Let $\nuI \in \mcAI$, let $\nuB,\nuB'\in\mcAB$ and $\nuBdir:=(\nuB'-\nuB)$, and let $\eps>0$. Next, for $t\in\mfT$ we have
\begin{equation}
\QB[B,\nuB+\eps\,\nuBdir]_t = \QB[B,\nuB]_t + \eps\,\int_0^t \nuBdir_u\, \diff u\,,
\end{equation}
and 
\begin{equation}
Y^{\nuB+\eps\,\nuBdir}_t = Y^\nuB_t + \eps\,\instantB\,\int_0^t e^{-\decayB\,(t-u)}\,\nuBdir_u\, \diff u\,,
\end{equation}
therefore,
\begin{align*}
\pcB(\nuB+\eps\,\nuBdir,\nuI) - \pcB(\nuB,\nuI)  &=\eps\,\E\Bigg[
    \int_0^T \Bigg\{-2\,\tempB\,\nuB_t\,\nuBdir_t + \alpha_t\,\int_0^t \nuBdir_u \diff u\\
    & \qquad\qquad - (2\,\termpB-\instantB)\,\left( \QB[B,\nuB]_t\,\nuBdir_t  + \nuB_t\,\int_0^t\nuBdir_u\,\diff u \right)\\
    &\qquad \qquad - \decayB\,\instantB\,\QB[B,\nuB]_t\,\int_0^t e^{-\decayB\,(t-s)}\,\nuBdir_s\,\diff s - \decayB\,Y_t\,\int_0^t\nuBdir_s\,\diff s \\
    &\qquad \qquad + 2\,\termpB\,\left(\nuI_t+\nuU_t\right)\,\int_0^t\nuBdir_u\,\diff u -2\,\runnB\,\QB[B,\nuB]_t\,\int_0^t\nuBdir_u\,\diff u\bigg\}\,\diff t 
    \Bigg] + \mathcal{O}\left(\eps^2\right)\,.
\end{align*}
After applying Fubini's theorem, it follows that 
\begin{align}
\langle \mathcal{D}\,\pcB(\nuB,\nuI),\nuBdir\rangle & = 
\E\Bigg[
    \int_0^T \nuBdir_t\,\bigg\{-2\,\tempB\,\nuB_t   - (2\,\termpB-\instantB)\,\QB[B,\nuB]_t  - \decayB\,\instantB\,\int_t^T \QB[B,\nuB]_u \,e^{-\decayB\,(u-t)}\,\diff u 
    \\
    &\hspace*{6.5em}  +  \int_t^T \left(  \alpha_u
    - (2\,\termpB-\instantB)\,\nuB_u
    -\decayB\,Y_u +2\,\termpB\,\left(\nuI_u+\nuU_u\right)
    -2\,\runnB\,\QB[B,\nuB]_u    
    \right)\diff u  
 \bigg\}\,\diff t
 \Bigg]\,,\nonumber
\end{align}
which concludes the proof.
\end{proof}

The following theorem characterises the best response of the informed trader as the solution to an FBSDE.

\begin{theorem}\label{thm: FBSDE Inf opt}
Let $\nuB\in\mcAB$. A control $\nuIstar\in\mcAI$ maximises the functional $\nuI\to \pcI(\nuB,\nuI)$ if and only if 
\begin{equation}
\begin{cases}\label{eq: FBSDE 1 Informed}
    & \diff \nuIstar_t = -\frac{1}{2\,\tempI}\left(\alpha_t - 2\,\runnI\,\QI[I,\nuIstar]_t + \instantB\,\nuB_t - \decayB\, Y_t\right)\diff t + \frac{1}{2\,\tempI}\,\diff \MI_t \,,
    \\[0.5em]
    & \nuIstar_T = - \frac{\termpI}{\tempI}\,\QI[I,\nuIstar]_T\,,
\end{cases}
\end{equation}
for a suitable $\mcFI$-martingale  $(\MI_t)_{t\in\mfT}$.
\end{theorem}

\begin{proof}
Let $\nuB\in\mcAB$. By Proposition \ref{prop: pcI concave}, $\nuI\to \pcI(\nuB,\nuI)$ is a strictly concave functional  and by Proposition \ref{prop: Gateaux pcI} it follows from \cite{ekeland1999convex} (Proposition 2.1 in Chapter 2) that  $\nuIstar\in \mcAI$ maximises the functional $\pcI$ if and only if 
\begin{equation}\label{eq: gateaux zero condition}
     \langle \mathcal{D}\,\pcI(\nuB,\nuIstar),(\nuI'-\nuIstar)\rangle = 0\,,\quad \forall \;\nuI'\in\mcAI\,.
\end{equation}
Next, we show that this happens if and only if $\nuIstar$ satisfies the FBSDE in \eqref{eq: FBSDE 1 Informed}.
\textbf{Necessity:}  Let $\nuIstar\in\mcAI$ maximise $\nuI\to\pcI(\nuB,\nuI)$, so \eqref{eq: gateaux zero condition} holds. Then, for  $\nuI'\in\mcAI$ arbitrary, from Proposition \ref{prop: Gateaux pcI} we  have
\begin{align}
\E\Bigg[\int_0^T \nuIdir_t\,\bigg\{ 
-2\,\tempI\,\nuIstar_t  -2\,\termpI\, \QI[I,\nuIstar]_t + \E\bigg[\int_t^T\chi_u^{\nuIstar}\diff u \,\bigg|\,\mcFI_t\bigg]
\bigg\}\,\diff t
\Bigg]=0\,,
\end{align}
where $\nuIdir:=(\nuI'-\nuI^*)$ and where $\chi_u^{\nuIstar}:=(\alpha_u
 + \instantB\,\nuB_u - \decayB\,{Y}_u  -2\,\termpI\,\nuIstar_u -2\,\runnI\,\QI[I,\nuIstar]_u)$.
This is only possible if 
\begin{align}
 2\,\tempI\,\nuIstar_t &=  
-2\,\termpI\, \QI[I,\nuIstar]_t + \E\bigg[\int_t^T\chi_u^{\nuIstar}\diff u \,\bigg|\,\mcFI_t\bigg]\,.
\end{align}
Next, define the $\mcFI$-martingale $(\MI_t)_{t\in\mfT}$ by
\begin{equation}
\MI_t := \E\bigg[\int_0^T\chi_u^{\nuIstar}\diff u \,\bigg|\,\mcFI_t\bigg]\,,
\end{equation}
for $t\in\mfT$. As $\alpha, Y,\nuB,\nuIstar,\QI[I,\nuIstar]$ are square-integrable and $\chi_u^{\nuIstar}$ is a linear combination of square-integrable adapted processes, we have that 
\begin{equation} 
\E\left[ \left| \int_0^T\chi_u^{\nuIstar}\diff u\right|\; \right]
\le\left(\E\left[\left(\int_0^T\chi_u^{\nuIstar} \,\diff u\right)^2\right]\right)^{\tfrac12}
\le 
\left( T\;\E\left[\int_0^T(\chi_u^{\nuIstar})^2 \,\diff u\right]\right)^{\tfrac12}
<+\infty   \,,
\end{equation}
hence, $(\MI_t)_{t\in\mfT}$ is a martingale. Thus
we obtain the representation
\begin{align}
 2\,\tempI\,\nuIstar_t &=  
-2\,\termpI\, \QI[I,\nuIstar]_t - \int_0^t\chi_u^{\nuIstar}\diff u + \MI_t\,,
\end{align}
and it follows that 
\begin{align}
 2\,\tempI\,\diff \nuIstar_t &=  -2\,\termpI\, \diff \QI[I,\nuIstar]_t -\chi_t^{\nuIstar}\diff t + \diff \MI_t\\
&=-\left(\alpha_t
 + \instantB\,\nuB_t - \decayB\,{Y}_t  -2\,\runnI\,\QI[I,\nuIstar]_t
\right)\diff t + \diff \MI_t\,.
\end{align}
\textbf{Sufficiency:} Let $\nuIstar\in \mcAI$ satisfy \eqref{eq: FBSDE 1 Informed}. The solution to the FBSDE can be written as
\begin{align}
 \nuIstar_t &=  -\frac{\termpI}{\tempI}\, \QI[I,\nuIstar]_t + \frac{1}{2\,\tempI}\E\bigg[\int_t^T \chi_u^{\nuIstar}\,\diff u \,\bigg|\,\mcFI_t\bigg]\,.
\end{align}
Then, for $\nuIdir:=(\nuI'-\nuI^*)$ with $\nuI'\in\mcAI$ arbitrary,
\allowdisplaybreaks
\begin{align*}
    \langle \mathcal{D}\,\pcI(\nuB,\nuIstar),\nuIdir\rangle &=  \E\left[\int_0^T \nuIdir_t \left\{ -2\,\tempI\,\nuIstar_t - 2\,\termpI\,\QI[I,\nuIstar]_t + \int_t^T \chi_u^{\nuIstar} \diff u\right\} \diff t\right]\nonumber\\
    &=  \E\Bigg[\int_0^T \nuIdir_t \bigg\{ 
    - \E\bigg[\int_t^T\chi_u^{\nuIstar}\diff u \,\bigg|\,\mcFI_t\bigg] + \int_t^T \chi_u^{\nuIstar} \diff u
    \bigg\}\,\diff t\Bigg]
    \\
 &=  \E\Bigg[\int_0^T \nuIdir_t \bigg\{ 
    - \E\bigg[\int_0^T \chi_u^{\nuIstar} \diff u \,\bigg|\,\mcFI_t\bigg]
    + \int_0^T \chi_u^{\nuIstar} \diff u
    \bigg\}\,\diff t\Bigg]\\    
&=  \E\Bigg[\int_0^T \nuIdir_t \bigg\{ 
    - \MI_t + \MI_T \bigg\}\,\diff t\Bigg] =  \E\Bigg[\int_0^T \nuIdir_t \bigg\{ 
    \E\bigg[- \MI_t + \MI_T \,\Big|\,\mcFI_t\bigg] \bigg\}\,\diff t\Bigg] = 0\,.
\end{align*}
Thus, $\nuIstar$ is optimal because
\begin{equation}
    \langle \mathcal{D}\,\pcI(\nuB,\nuIstar),(\nuI'-\nuI^*)\rangle = 0\,,\qquad \forall \nuI'\in\mcAI \qquad \iff \quad \nuIstar = \argmax_{\nuI\in\mcAI} \pcI(\nuI,\nuBstar)\,.
\end{equation}
\end{proof}

Similar to the theorem above, the theorem below characterises the best response of the broker  as the solution to an FBSDE.

\begin{theorem}\label{thm: FBSDE Broker opt}
Let $\nuI\in\mcAI$. A control $\nuBstar\in\mcAB$ maximises the functional $\nuB\to \pcB(\nuB,\nuI)$ if and only if
\begin{equation}\label{eq: FBSDE 2 Broker}
\begin{cases}
    &\diff \nuBstar_t  = - \frac{1}{2\,\tempB} \left(  \alpha_t
    - \left(2\,\runnB + \decayB\,\instantB \right)\QB[B,\nuBstar]_t    
    -\decayB\,Y_t 
    + \instantB\,\left(\nuI_t + \nuU_t\right)
    + \decayB^2\,\instantB\,Z_t 
    \right)\diff t 
    + \frac{1}{2\,\tempB} \diff \MBb_t - \frac{\decayB\,\instantB}{2\,\tempB}e^{\decayB\,t}\,\diff \MBa_t,
    \\[0.5em]
    & \diff Z_t = \left( \decayB\,Z_t  - \QB[B,\nuBstar]_t\right)\diff t + e^{\decayB\,t}\diff \MBa_t\,, 
    \\[0.5em]
    & Z_T = 0 \,,
    \\[0.5em]
    & \nuBstar_T = - \frac{2\,\termpB - \instantB}{2\,\tempB}\,\QB[B,\nuBstar]_T\,,
\end{cases}
\end{equation}
for suitable $\mcFB$-martingales  $(\MBa_t)_{t\in\mfT}$ and $(\MBb_t)_{t\in\mfT}$.
\end{theorem}

\begin{proof}
The proof is similar to that of Theorem \ref{thm: FBSDE Inf opt}.
Let $\nuI\in\mcAI$. By Proposition \ref{prop: pcB concave}, $\nuB\to \pcB(\nuB,\nuI)$ is a strictly concave functional and as before, by Proposition \ref{prop: Gateaux pcB}, it follows  that  $\nuBstar\in \mcAB$ maximises the functional $\pcB$ if and only if 
\begin{equation}\label{eq: gateaux B zero condition}
     \langle \mathcal{D}\,\pcB(\nuBstar,\nuI), (\nuB'-\nuB^*)\rangle = 0\,,\quad \forall \nuB'\in\mcAB\,.
\end{equation}
Next, we show that this happens if and only if $\nuBstar$ satisfies the FBSDE in \eqref{eq: FBSDE 2 Broker}.
\textbf{Necessity:}  Let $\nuBstar\in\mcAB$ maximise $\nuB\to\pcB(\nuB,\nuI)$, then, it follows that  \eqref{eq: gateaux B zero condition} holds. Take $\nuB'\in\mcAB$ arbitrary, then from Proposition \ref{prop: Gateaux pcB}, we have 
\begin{align}
0 &= \E\Bigg[
    \int_0^T \nuBdir_t\,\bigg\{-2\,\tempB\,\nuBstar_t   - (2\,\termpB-\instantB)\,\QB[B,\nuBstar]_t  - \decayB\,\instantB\,\int_t^T \QB[B,\nuBstar]_u \,e^{-\decayB\,(u-t)}\,\diff u   +  \int_t^T \chi_u^{\nuBstar}\diff u  
 \bigg\}\,\diff t
 \Bigg]\,,\nonumber
\end{align}
where $\nuBdir:=(\nuB'-\nuB^*)$ and $\chi_u^{\nuBstar}:=\left(  \alpha_u
    - (2\,\termpB-\instantB)\,\nuBstar_u
    -\decayB\,Y_u +2\,\termpB\,\left(\nuI_u+\nuU_u\right)
    -2\,\runnB\,\QB[B,\nuBstar]_u    
    \right)$. 
    This is only possible if 
\begin{align}\label{eq: aux necessity broker}
2\,\tempB\,\nuBstar_t  &= - (2\,\termpB-\instantB)\,\QB[B,\nuBstar]_t  - \decayB\,\instantB\,e^{\decayB\,t}\,\E\left[\int_t^T \QB[B,\nuBstar]_u \,e^{-\decayB\,u}\,\diff u \,\mid\,\mcFB_t \right]  + \E\left[ \int_t^T \chi_u^{\nuBstar}\diff u    \,\mid\,\mcFB_t \right]\,.\nonumber
\end{align}
Hence, define the $\mcFB$-martingales $(\MBa_t)_{t\in\mfT}$ and $(\MBb_t)_{t\in\mfT}$ s.t.,
\begin{align}
\MBa_t := \E\left[\int_0^T \QB[B,\nuBstar]_u \,e^{-\decayB\,u}\,\diff u \,\mid\,\mcFB_t \right]
\qquad \text{and} \qquad
\MBb_t  := \E\left[\int_0^T\chi_u^{\nuBstar}\diff u    \,\mid\,\mcFB_t \right]\,,
\end{align}
for $t\in\mfT$. Similar to the proof that $\MI$ is a martingale in Theorem \ref{thm: FBSDE Inf opt}, $\MBa$ and $\MBb$ are martingales because $\QB[B,\nuBstar]$ and $\chi_u^{\nuBstar}$ are square-integrable. In particular, we have that 
\begin{equation}
\E\left[\left|\int_0^T\QB[B,\nuBstar]_u \,e^{-\decayB\,u}\diff u\right|\right]<+\infty
\qquad \text{and} \qquad 
\E\left[\left|\int_0^T\chi_u^{\nuBstar}\diff u\right|\right]<+\infty   \,.
\end{equation}

Hence, we obtain the representation 
\begin{align}
2\,\tempB\,\nuBstar_t  &= - (2\,\termpB-\instantB)\,\QB[B,\nuBstar]_t  - \decayB\,\instantB\,e^{\decayB\,t}\,\left(\MBa_t - \int_0^t \QB[B,\nuBstar]_u \,e^{-\decayB\,u}\,\diff u  \right)  - \int_0^t \chi_u^{\nuBstar}\diff u  + \MBb_t\,.
\end{align}
Next, define the auxiliary process 
\begin{equation}
    Z_t = e^{\decayB\,t}\,\left(\MBa_t - \int_0^t \QB[B,\nuBstar]_u \,e^{-\decayB\,u}\,\diff u  \right)\,,
\end{equation}
which satisfies the BSDE
\begin{equation}
    \diff Z_t = \left( \decayB\,Z_t  - \QB[B,\nuBstar]_t\right)\diff t + e^{\decayB\,t}\diff \MBa_t\,,\qquad Z_T =0\,.
\end{equation}
Thus, 
\begin{align}
2\,\tempB\,\diff \nuBstar_t  &= - (2\,\termpB-\instantB)\,\left(\nuBstar_t - \nuI_t - \nuU_t\right)\diff t  - \decayB\,\instantB\,\left[\left( \decayB\,Z_t  - \QB[B,\nuBstar]_t\right)\diff t + e^{\decayB\,t}\diff \MBa_t\right]   - \chi_t^{\nuBstar}\diff t  + \diff \MBb_t\,,
\end{align}
which simplifies to 
\begin{align}
\diff \nuBstar_t  &= - \frac{1}{2\,\tempB} \left(  \alpha_t
    - \left(2\,\runnB + \decayB\,\instantB \right)\QB[B,\nuBstar]_t    
    -\decayB\,Y_t 
    + \instantB\,\left(\nuI_t + \nuU_t\right)
    + \decayB^2\,\instantB\,Z_t 
    \right)\diff t  + \frac{1}{2\,\tempB} \diff \MBb_t - \frac{\decayB\,\instantB}{2\,\tempB}e^{\decayB\,t}\diff \MBa_t\,,
\end{align}
and together with \eqref{eq: aux necessity broker} evaluated at $T$ show that the FBSDE system is satisfied.
\textbf{Sufficiency:} Let $\nuBstar\in \mcAB$ satisfy \eqref{eq: FBSDE 2 Broker}. The  solution to the FBSDE can be written as
\begin{align}
\nuBstar_t  &= - \frac{(2\,\termpB-\instantB)}{2\,\tempB}\,\QB[B,\nuBstar]_t  - \frac{\decayB\,\instantB\,e^{\decayB\,t}}{2\,\tempB}\,\E\left[\int_t^T \QB[B,\nuBstar]_u \,e^{-\decayB\,u}\,\diff u \,\mid\,\mcFB_t \right] + \frac{1}{2\,\tempB} \E\left[ \int_t^T \chi_u^{\nuBstar}\diff u    \,\mid\,\mcFB_t \right]\,,\nonumber
\end{align}
which,  using Proposition  \ref{prop: Gateaux pcB} yields,
\begin{align*}
    \langle \mathcal{D}\,\pcB(\nuBstar,\nuI),\nuBdir\rangle &=  \E\Bigg[\int_0^T \nuBdir_t \bigg\{ \decayB\,\instantB\,e^{\decayB\,t}\left( \MBa_t -\MBa_T \right)
    - \MBb_t + \MBb_T \bigg\}\,\diff t\Bigg] \\
    &=  \E\Bigg[\int_0^T \nuBdir_t \bigg\{ 
    \decayB\,\instantB\,e^{\decayB\,t}\E\bigg[ \MBa_t -\MBa_T \,\Big|\,\mcFI_t\bigg]
    + \E\bigg[ \MBb_T - \MBb_t 
    \,\Big|\,\mcFI_t\bigg] \bigg\}\,\diff t\Bigg] = 0\,,
\end{align*}
for all $\nuB'\in\mcAB$ and $\nuBdir:=\nuB'-\nuB^*$. Thus, 
it follows that $\nuBstar$ is optimal. 
\end{proof}

The next theorem shows that if there is a solution to a coupled system of FBSDEs, then, the solution to the system satisfies Definition \ref{def: Nash equilibrium}, and thus, it is a Nash equilibrium of the stochastic game.

\begin{theorem}
If there is $\nuBstar\in\mcAB$ and $\nuIstar\in\mcAI$ that satisfy the system of FBSDEs
\begin{subequations}\label{eqns:system-FBSDEs-orig varaibles}
\begin{align}
%%% BROKER %%%
    &\diff \nuBstar_t  = - \frac{1}{2\,\tempB} \left(  \alpha_t
    - \left(2\,\runnB + \decayB\,\instantB \right)\QB[B,\nuBstar]_t    
    -\decayB\,Y_t 
    + \instantB\,\left(\nuIstar_t + \nuU_t\right)
    + \decayB^2\,\instantB\,Z_t 
    \right)\diff t \nonumber
    \\
    &\hspace{8cm} + \frac{1}{2\,\tempB} \diff \MBb_t - \frac{\decayB\,\instantB}{2\,\tempB}e^{\decayB\,t}\,\diff \MBa_t\,,\\
%%% BROKER TERMINAL CONDITION %%%    
& \nuBstar_T = - \frac{\varphi}{\tempB}\,\QB[B,\nuBstar]_T\,,\\
%%% INFORMED %%%        
& \diff \nuIstar_t = -\frac{1}{2\,\tempI}\left(\alpha_t - 2\,\runnI\,\QI[I,\nuIstar]_t +  \instantB\,\nuBstar_t - \decayB\, Y_t \right)\diff t + \diff \MI_t\,,\\
%%% INFORMED TERMINAL CONDITION %%%
& \nuIstar_T = - \frac{\termpI}{\tempI}\,\QI[I,\nuIstar]_T\,,\\
%%% Y-process %%%    
&\diff Y_t = ( \instantB\,\nuBstar_t-\decayB\,Y_t )\,\diff t  \,,\quad Y_0\in\mathbb{R}\,,\\
%%% Z-process %%%    
& \diff Z_t = \left( \decayB\,Z_t  - \QB[B,\nuBstar]_t\right)\diff t + e^{\decayB\,t}\diff \MBa_t  \,,\qquad Z_T = 0 \,,\\
%%% inventory-broker %%%    
&\diff \QB[B,\nuBstar]_t = \left(\nuBstar_t - \nuIstar_t -\nuU_t\right)\diff t\,,\qquad \QB[B,\nuBstar]_0 =q^B_0\,,\\
%%% inventory-informed %%%    
&\diff \QI[I,\nuIstar]_t = \nuIstar_t\,\diff t\,,\qquad \QI[I,\nuIstar]_0 =q^I_0\,,
\end{align}
\end{subequations}
where  $\varphi = \termpB - \tfrac12\instantB$, and  the processes $(\MBa_t)_{t\in\mfT}$, $(\MBb_t)_{t\in\mfT}$, and $(\MI_t)_{t\in\mfT}$ are $\mcF$-martingales. Then, $(\nuBstar,\nuIstar)\in\mcAB\times \mcAI$ is a Nash equilibrium of the trading strategies between the broker and the informed trader.
\end{theorem}

\begin{proof}
Suppose there is $\nuIstar\in\mcAI$ and $\nuBstar\in\mcAB$ such that the system of FBSDEs  \eqref{eqns:system-FBSDEs-orig varaibles} holds;
then, by Theorem \ref{thm: FBSDE Inf opt} it follows that 
\begin{equation}
    \pcI(\nuIstar,\nuBstar) \geq \pcI(\nuIbis,\nuBstar)\,,\qquad \forall\;\nuIbis\in \mcAI\,.
\end{equation}
Similarly, by Theorem \ref{thm: FBSDE Broker opt}, it follows that
\begin{equation}
\pcB(\nuIstar,\nuBstar) \geq \pcB(\nuIstar,\nuBbis)\,,\qquad \forall\;  \nuBbis\in \mcAB \,.
\end{equation}
Thus, $(\nuIstar,\nuBstar)$ are a Nash equilibrium because they satisfy Definition \ref{def: Nash equilibrium}. 
\end{proof}

{\new{In general, existence and uniqueness of the FBSDE do not imply uniqueness of the Nash equilibrium. Here, given that for any $\nuI\in\mcAI$ we have that $\nuB\to \pcB(\nuB,\nuI)$ is strictly concave and for any $\nuB\in\mcAB$ we have that $\nuI\to \pcI(\nuB,\nuI)$ is strictly concave, and given Theorems \ref{thm: FBSDE Inf opt} and \ref{thm: FBSDE Broker opt}, existence and uniqueness of the FBSDE indeed implies uniqueness of the Nash equilibrium. If there is another Nash equilibrium $(\tilde{\nuB}, \tilde{\nuI})$, then, for $\tilde{\nuB}\in\mcAB$, Theorem \ref{thm: FBSDE Inf opt} implies that $\tilde{\nuI}$ satisfies the FBSDE of the informed trader, similarly,  for $\tilde{\nuI}\in\mcAI$, Theorem \ref{thm: FBSDE Broker opt} implies that $\tilde{\nuB}$ satisfies the FBSDE of the broker. Given that the two of them are satisfied simultaneously, then $(\tilde{\nuB}, \tilde{\nuI})$ is a solution to the coupled FBSDE. Given uniqueness of the coupled FBSDE we then have that the Nash equilibrium $(\tilde{\nuB}, \tilde{\nuI})$ coincides with $(\nuBstar, \nuIstar)$.  }}

\subsection{Existence and uniqueness}
In this section we are concerned with the existence and uniqueness of the system of FBSDEs above  that characterises the Nash equilirbium.
Introduce  the (vector valued) forward process $(\fbsdeX_t)_\tT$, backward process $(\fbsdeY_t)_\tT$, and martingale process $(\fbsdeM_t)_\tT$, where
\begin{align}
\fbsdeX_t &= 
\begin{pmatrix}
    \QB[B,\nuBstar]_t\\
    \QI[I,\nuIstar]_t\\
    Y_t
\end{pmatrix}\,,\qquad 
    \fbsdeY_t = 
\begin{pmatrix}
    \nuBstar_t\\
    \nuIstar_t\\
    Z_t
\end{pmatrix}\,,\qquad \text{and} \qquad
\fbsdeM_t = 
\begin{pmatrix}
    \MBa_t\\
    \MBb_t\\
    \MI_t
\end{pmatrix}
\,.
\end{align}
The FBSDE in \eqref{eqns:system-FBSDEs-orig varaibles} can be written  as follows
\begin{subequations}
    \begin{align}
\diff \fbsdeX_t &= \left( \fbsdeA\,\fbsdeX_t + \fbsdeB\,\fbsdeY_t  +\fbsdeb_t\right)\,\diff t\,,
& \fbsdeX_0 & = \left( q^B_0, q^I_0, Y_0 \right)^{\intercal},
\label{eqn: sde-for-X}
\\
\diff \fbsdeY_t &= \left( \fbsdeAhat\,\fbsdeX_t + \fbsdeBhat\,\fbsdeY_t + \fbsdebhat_t \right)\,\diff t + \fbsdesigma_t\,\diff \fbsdeM_t \,,\label{bsde of Y}
& \fbsdeY_T & = \fbsdeG\,\fbsdeX_T \,,
\end{align}%
\label{eqn:fbsde-system}%
\end{subequations}%
where $\fbsdeA,\fbsdeB,\fbsdeAhat, \fbsdeBhat, \fbsdeG$ are deterministic matrices in $\R^{3\times 3}$, the stochastic processes $\fbsdeb_t,\fbsdebhat_t:\mfT\times \Omega\to \R^3$ are  valued in $\R^3$, and $\fbsdesigma_t:
\mfT\to\R^{3\times 3}$ is a deterministic function, all of which we obtain from  \eqref{eqns:system-FBSDEs-orig varaibles}.
More precisely,
\allowdisplaybreaks
{
\begin{gather}
\fbsdeA = \begin{pmatrix}
       \, 0 & 0 & 0 \vspace{0.1cm}\\
       0 & 0 & 0 \vspace{0.1cm}\\
       0 & 0 & -\decayB
\end{pmatrix},
\qquad
\fbsdeB = \begin{pmatrix}
       1 & -1  &  0\vspace{0.1cm}\\
       0 & 1 & 0 \vspace{0.1cm}\\
       \instantB & 0 & 0
\end{pmatrix},
\qquad
\fbsdeb_t = \begin{pmatrix}
     -\nuU_t\\
     0\\
     0
\end{pmatrix},
\vspace{0.5cm}
\\
\fbsdeAhat = \begin{pmatrix}
\frac{2\,\runnB+\decayB\,\instantB}{2\tempB} & 0 & \frac{\decayB}{2\,\tempB} \vspace{0.1cm}
\\
0 &  \frac{\runnI}{\tempI} & \frac{\decayB}{2\,\tempI} \vspace{0.1cm}\\
-1 & 0  & 0
\end{pmatrix},
\qquad 
\fbsdeBhat = \begin{pmatrix}
0 & -\frac{\instantB}{2\,\tempB} & -\frac{\decayB^2 \,\instantB  }{2\, \tempB} \vspace{0.1cm}\\
-\frac{\instantB}{2\,\tempI} & 0 & 0 \vspace{0.1cm}\\
0 & 0 & \decayB\,
\end{pmatrix},
\\
\fbsdeG = \begin{pmatrix}
- \frac{\varphi}{\tempB} & 0 & 0  \\
0 &  - \frac{\termpI}{\tempI} & 0 \\
0 & 0 & 0      
\end{pmatrix},
\qquad 
\fbsdebhat_t = \begin{pmatrix}
- \frac{\alpha_t + \instantB\,\nuU_t}{2\,\tempB} \vspace{0.1cm}\\
- \frac{\alpha_t}{2\,\tempI} \vspace{0.1cm}\\
0
\end{pmatrix},
\qquad 
\fbsdesigma_t = \begin{pmatrix}
- \frac{\decayB\,\instantB}{2\,\tempB}e^{\decayB\,t} & \frac{1}{2\,\tempB} & 0 \vspace{0.1cm}\\
0 & 0 & 1 \vspace{0.1cm}\\
e^{\decayB\,t} & 0 & 0 
\end{pmatrix}.
\end{gather}
}

{\new{There are a number of theorems in the literature that show existence and uniqueness of FBSDEs.
If $T$ is small enough, existence and uniqueness was first studied in \cite{antonelli1993backward}; for an overview see  Chapter 2 in \cite{carmona2016lectures} and the references therein. Beyond the small time horizon,
in the Brownian case,  }}
\cite{yong1999linear} shows that existence and uniqueness of a solution to the FBSDE system \eqref{eqn:fbsde-system} is equivalent to the existence and uniqueness of a matrix Riccati differential equation; this result follows the  `four-step scheme' of \cite{ma1994solving}.\footnote{ See \cite{ma1999forward,carmona2016lectures} for manuscripts covering a wide range of existence and uniqueness results for FBSDEs. }

{\new{As noted in page 58 in \cite{carmona2016lectures} regarding the affine FBSDE (2.50):  \emph{``except for possibly in the case $C_t=D_t=0$ and $\sigma_t\equiv\sigma>0$,  and despite the innocent-looking form of the coefficients of FBSDE (2.50),  none of the existence theorems which we \textit{slaved} to prove in the previous sections apply to this simple particular case''}. If we were to restrict $\MI,\MBa,\MBb$ to be Brownian motions, our FBSDE would be as in (2.50) in \cite{carmona2016lectures} but for $C_t=D_t=\sigma_t=0$ (not $\sigma_t\equiv\sigma>0$).  Our FBSDE goes beyond the Brownian context. Given that we did not find a short-time-horizon existence and uniqueness result that was readily applicable in our case, we show, for the sake of completeness, the precise constraint on $T$ that guarantees existence and uniqueness in the present context. }}

{\new{Below, we prove existence and uniqueness of the above FBSDE under three different set of restrictions on model parameters.  
Firstly, in the general case, we find an explicit bound on $T$ that guarantees existence and uniqueness.
Secondly, when the values of the parameters $\decayB$ and $\runnB$ are zero, existence and uniqueness of the solution follows without any further conditions on the remaining model parameters (in particular, no restrictions on $T$). 
Thirdly, if only $\decayB$ is zero, existence and uniqueness follows under four additional (mild) restrictions on model parameters (namely, $\instantB \leq 2\,\tempI,\,2\,\tempB,\,4\,\runnI,\,4\,\runnB$). Out of the three proofs that we present (small $T$, $\decayB=\runnB=0$, and $\decayB=0$ with $\runnB>0$) the first proof is new and follows the ideas in \cite{antonelli1993backward}. The second and third are standard because we show existence and uniqueness at the level of the matrix Riccati differential equation.
}}

{\new{We start by finding the explicit bound on $T$ so that there is a unique solution to the above FBSDE.}} Let 
\begin{equation}
  \S^2 := \Bigg\{ (\gamma_t)_{t\in\mfT} \,\,|\,\, \gamma\text{ is } \mcFB-\text{adapted, and } \E\left[\sup_{t\in\mfT}\textstyle|\gamma_t|^2\right]\,<\,+\infty \Bigg\}  \,,
\end{equation}
and define 
$\spaceSSS =\S^2\times\S^2\times\S^2 $ and let $\spaceK = \spaceSSS\times \spaceSSS$. 
We endow the 
vector space $\spaceK$ with the following  norm: for $(X,Y)\in\spaceK$ 
\begin{equation}
\norm{(X,Y)}_{\spaceK}^2 :=  \mathbb{E}\left[\sup_{t\in\mfT}  |{X_t}|^2\right]  +  \mathbb{E}\left[\sup_{t\in\mfT}|{Y_t}|^2 \right]\,,
\end{equation}
where  $|\cdot|$  is the Euclidean norm in $\R^3$. 
When $|{\cdot}|$ is applied to a matrix $M\in\R^{3\times 3}$, then
\begin{equation}
    |{M}| := \sup_{x\in\R^3,\, 0<|{x}|\leq 1}\frac{|{M\,x}|}{|{x}|}\,.
\end{equation}

The theorem below shows the bound that $T$ needs to satisfy so that existence and uniqueness follows from a fixed-point argument.

\begin{theorem}\label{thm: existence uniqueness}
If the time horizon $T$ and $|\fbsdeG|$ are such that  
\begin{equation}
  \max\left\{ 12\,|\fbsdeG|^2 + T^2(2\,|\fbsdeA|^2+30\,|\fbsdeAhat|^2)\,,\,  T^2\left(2\,|\fbsdeB|^2 + 30\,|\fbsdeBhat|^2\right) \right\}<1\,,  
\end{equation}
then, there is a unique solution to the FBSDE system in \eqref{eqns:system-FBSDEs-orig varaibles}. When $|\fbsdeG| = 0$, the inequality reduces to \begin{equation}\label{eq:bound-on-T}
    T < 1 \bigg/ \sqrt{ \left( \max\left\{  2\,|\fbsdeA|^2 + 30\,|\fbsdeAhat|^2,\,  2\,|\fbsdeB|^2 + 30\,|\fbsdeBhat|^2 \right\} \right)
    }\,.
\end{equation}
\end{theorem}

\begin{proof}
Consider the mapping $\Gamma:\spaceK\to \spaceK$ that maps $(\fbsdeX,\fbsdeY)\in\spaceK$ to
\begin{align*}
    \Gamma(\fbsdeX,\fbsdeY)_t = \left(\fbsdeX_0+\int_0^t \left( \fbsdeA\,\fbsdeX_s + \fbsdeB\,\fbsdeY_s  + \fbsdeb_s\right)\diff s,\,
    \mathbb{E}\left[\fbsdeG\,\fbsdeX_T - \int_t^T\left( \fbsdeAhat\,\fbsdeX_s + \fbsdeBhat\,\fbsdeY_s + \fbsdebhat_s \right)\diff s \,\Big| \,\mcF_t \right] \right)\,.
\end{align*}
First, we show that indeed for $(\fbsdeX,\fbsdeY)\in\spaceK$, we have that $\Gamma(\fbsdeX,\fbsdeY)\in\spaceK$. Observe that 
\begin{align}
\norm{\Gamma(\fbsdeX,\fbsdeY)}_{\spaceK}^2 &= \underbrace{\mathbb{E}\left[ \sup_{t\in\mfT} \left|{\fbsdeX_0+\int_0^t \left( \fbsdeA\,\fbsdeX_s + \fbsdeB\,\fbsdeY_s  +\fbsdeb_s\right)\,\diff s}\right|^2 \right] }_{=:\mfA}\\
&\quad +  \underbrace{\mathbb{E}\left[ \sup_{t\in\mfT}\Bigg|{\mathbb{E}\left[\fbsdeG\fbsdeX_T - \int_t^T\left( \fbsdeAhat\,\fbsdeX_s + \fbsdeBhat\,\fbsdeY_s + \fbsdebhat_s \right)\,\diff s \,\Big| \,\mcF_t \right]}\Bigg|^2 \diff t \right]}_{=:\mfB}\,.
\end{align}
For the first term, using the triangle inequality and Cauchy-Schwarz inequality we find the upper bound 
\begin{align}
\mfA &\leq 2\,\mathbb{E}\left[ \sup_{t\in\mfT} |\fbsdeX_t|^2\right] + 2\,T\,\mathbb{E}\left[ \sup_{t\in\mfT}  \int_0^t \left| \fbsdeA\,\fbsdeX_s + \fbsdeB\,\fbsdeY_s  +\fbsdeb_s\right|^2\,\diff s\right]\\
&\leq 2\,\mathbb{E}\left[ \sup_{t\in\mfT} |\fbsdeX_t|^2\right] + 2\,T\,\mathbb{E}\left[ \int_0^T \left| \fbsdeA\,\fbsdeX_s + \fbsdeB\,\fbsdeY_s  +\fbsdeb_s\right|^2\,\diff s\right]\\
& \leq 2\,\mathbb{E}\left[ \sup_{t\in\mfT} |\fbsdeX_t|^2\right] + 6\,T^2\,\mathbb{E}\left[ \sup_{t\in\mfT}   |\fbsdeA|^2\,|\fbsdeX_t|^2 + |\fbsdeB|^2\,|\fbsdeY_t|^2  + |\fbsdeb_t|^2\right]\,,
\end{align}
which is finite because $(\fbsdeX,\fbsdeY)\in\spaceK$.
For the second term, using Cauchy-Schwarz inequality we find the following upper bound
\begin{align}
\mfB &\leq 3\,\mathbb{E}\left[ \sup_{t\in\mfT} \left|\mathbb{E}\left[\fbsdeG\,\fbsdeX_T\,\Big| \,\mcF_t\right]\right|^2\right] + 3\, \mathbb{E}\left[ \sup_{t\in\mfT} \left|\mathbb{E}\left[\int_0^T\left( \fbsdeAhat\,\fbsdeX_s + \fbsdeBhat\,\fbsdeY_s + \fbsdebhat_s \right)\,\diff s\,\Big| \,\mcF_t\right]\right|^2\right] \\
&\qquad + 3\, \mathbb{E}\left[ \sup_{t\in\mfT} \left|\int_0^t\left( \fbsdeAhat\,\fbsdeX_s + \fbsdeBhat\,\fbsdeY_s + \fbsdebhat_s \right)\,\diff s\right|^2\right] \,.
\end{align}
The third term above can be bound similarly to the bound for $\mfA$. For the first two terms above, we use Doob's martingale inequality to find that 
\begin{align}
&\mathbb{E}\left[ \sup_{t\in\mfT} \left|\mathbb{E}\left[\fbsdeG\,\fbsdeX_T\,\Big| \,\mcF_t\right]\right|^2\right] + \mathbb{E}\left[ \sup_{t\in\mfT} \left|\mathbb{E}\left[\int_0^T\left( \fbsdeAhat\,\fbsdeX_s + \fbsdeBhat\,\fbsdeY_s + \fbsdebhat_s \right)\,\diff s\,\Big| \,\mcF_t\right]\right|^2\right]\\
&\qquad \leq   4\, \mathbb{E}\left[ \left|\fbsdeG\,\fbsdeX_T\right|^2\right] + 4\, \mathbb{E}\left[ \left|\int_0^T\left( \fbsdeAhat\,\fbsdeX_s + \fbsdeBhat\,\fbsdeY_s + \fbsdebhat_s \right)\,\diff s\right|^2\right]
\\
&\qquad \leq   4\,|\fbsdeG|^2\, \mathbb{E}\left[ \sup_{t\in\mfT}|\fbsdeX_t|^2\right] 
+ 4\, \mathbb{E}\left[ \left|\int_0^T\left( \fbsdeAhat\,\fbsdeX_s + \fbsdeBhat\,\fbsdeY_s + \fbsdebhat_s \right)\,\diff s\right|^2\right]\,,
\end{align}
from which we conclude that $\mfB$ is finite. 
Thus, $\norm{\Gamma(\fbsdeX,\fbsdeY)}_{\spaceK}^2 <+\infty$ and it follows that  $\Gamma(\fbsdeX,\fbsdeY)\in\spaceK$ because $\Gamma(\fbsdeX,\fbsdeY)$ is adapted by construction.

Next, we show that for $T$ small enough $\Gamma$ is a contraction mapping.
Let $(\fbsdeX,\fbsdeY)\in\spaceK$ and $(\tilde\fbsdeX,\tilde\fbsdeY)\in\spaceK$, it follows that
\begin{align}
    &\norm{\Gamma(\fbsdeX,\fbsdeY) - \Gamma(\tilde\fbsdeX,\tilde\fbsdeY)}^2_{\spaceK} \\
    &\quad= \mathbb{E}\left[ \sup_{t\in\mfT} \left|{\int_0^t \fbsdeA\,\left(\fbsdeX_s - \tilde\fbsdeX_s\right) + \fbsdeB\,\left(\fbsdeY_s - \tilde\fbsdeY_s \right)\diff s }\right|^2\right]\\
    &\quad\qquad + \mathbb{E}\left[\sup_{t\in\mfT}  \left| \mathbb{E}\left[ \fbsdeG\,\left(\fbsdeX_T -\tilde\fbsdeX_T\right) - {\int_t^T  \fbsdeAhat\,\left(\fbsdeX_s - \tilde\fbsdeX_s\right) + \fbsdeBhat\,\left(\fbsdeY_s - \tilde\fbsdeY_s \right) \diff s }  \,\Big| \,\mcF_t \right]\right|^2  \right]\,.
\end{align}
Using  the triangle inequality, the Cauchy-Schwartz inequality, and Doob's martingale inequality, we find that 
\begin{subequations}
    \begin{align}
    &\norm{\Gamma(\fbsdeX,\fbsdeY) - \Gamma(\tilde\fbsdeX,\tilde\fbsdeY)}^2_{\spaceK} \\
&\qquad\leq  2\,T^2\,|\fbsdeA|^2\,\mathbb{E}\left[ \sup_{t\in\mfT} \left|\fbsdeX_t - \tilde\fbsdeX_t\right|^2 \right] 
    + 2\,T^2\,|\fbsdeB|^2\,\mathbb{E}\left[ \sup_{t\in\mfT} \left|\fbsdeY_t - \tilde\fbsdeY_t\right|^2 \right] 
    \\
    &\qquad\qquad + 3\,\mathbb{E}\left[\sup_{t\in\mfT}  \left| \mathbb{E}\left[ \fbsdeG\,\left(\fbsdeX_T -\tilde\fbsdeX_T\right)   \,\Big| \,\mcF_t \right]\right|^2  \right]
    \\
    &\qquad\qquad 
    + 3\,\mathbb{E}\left[\sup_{t\in\mfT}  \left| \mathbb{E}\left[  {\int_0^T  \fbsdeAhat\,\left(\fbsdeX_s - \tilde\fbsdeX_s\right) + \fbsdeBhat\,\left(\fbsdeY_s - \tilde\fbsdeY_s \right) \diff s }  \,\Big| \,\mcF_t \right]\right|^2  \right]\\
    &\qquad\qquad  + 3\,\mathbb{E}\left[\sup_{t\in\mfT}  \left|  \int_0^t  \fbsdeAhat\,\left(\fbsdeX_s - \tilde\fbsdeX_s\right) + \fbsdeBhat\,\left(\fbsdeY_s - \tilde\fbsdeY_s \right) \diff s  \right|^2  \right]
    \\
&\qquad\leq  2\,T^2\,|\fbsdeA|^2\,\mathbb{E}\left[ \sup_{t\in\mfT} \left|\fbsdeX_t - \tilde\fbsdeX_t\right|^2 \right] 
    + 2\,T^2\,|\fbsdeB|^2\,\mathbb{E}\left[ \sup_{t\in\mfT} \left|\fbsdeY_t - \tilde\fbsdeY_t\right|^2 \right] 
    \\
    %Doob 1
    &\qquad\qquad + 12\,| \fbsdeG|^2\,\mathbb{E}\left[\left| \fbsdeX_T -\tilde\fbsdeX_T   \right|^2  \right]
    \\
    %Doob 2
    &\qquad\qquad 
    + 12\,\mathbb{E}\left[ \left| \int_0^T  \fbsdeAhat\,\left(\fbsdeX_s - \tilde\fbsdeX_s\right) + \fbsdeBhat\,\left(\fbsdeY_s - \tilde\fbsdeY_s \right) \diff s  \right|^2  \right]
    \\
    &\qquad\qquad +  6\,T^2\,|\hat\fbsdeA|^2\,\mathbb{E}\left[ \sup_{t\in\mfT} \left|\fbsdeX_t - \tilde\fbsdeX_t\right|^2 \right] 
    + 6\,T^2\,|\hat\fbsdeB|^2\,\mathbb{E}\left[ \sup_{t\in\mfT} \left|\fbsdeY_t - \tilde\fbsdeY_t\right|^2 \right] \,.
\\
&\qquad\leq  2\,T^2\,|\fbsdeA|^2\,\mathbb{E}\left[ \sup_{t\in\mfT} \left|\fbsdeX_t - \tilde\fbsdeX_t\right|^2 \right] 
    + 2\,T^2\,|\fbsdeB|^2\,\mathbb{E}\left[ \sup_{t\in\mfT} \left|\fbsdeY_t - \tilde\fbsdeY_t\right|^2 \right] \\
    %Doob 1
    &\qquad\qquad + 12\,| \fbsdeG|^2\,\mathbb{E}\left[\sup_{t\in\mfT}\left| \fbsdeX_t -\tilde\fbsdeX_t   \right|^2  \right]\\
    %Doob 2
    &\qquad\qquad 
    + 24\,T^2\,|\hat\fbsdeA|^2\,\mathbb{E}\left[ \sup_{t\in\mfT} \left|\fbsdeX_t - \tilde\fbsdeX_t\right|^2 \right] 
    + 24\,T^2\,|\hat\fbsdeB|^2\,\mathbb{E}\left[ \sup_{t\in\mfT} \left|\fbsdeY_t - \tilde\fbsdeY_t\right|^2 \right]
    \\
    &\qquad\qquad +  6\,T^2\,|\hat\fbsdeA|^2\,\mathbb{E}\left[ \sup_{t\in\mfT} \left|\fbsdeX_t - \tilde\fbsdeX_t\right|^2 \right] 
    + 6\,T^2\,|\hat\fbsdeB|^2\,\mathbb{E}\left[ \sup_{t\in\mfT} \left|\fbsdeY_t - \tilde\fbsdeY_t\right|^2 \right] 
    \\
    &\qquad = (12\,|\fbsdeG|^2 + T^2\,(2\,|\fbsdeA|^2 + 30\,|\fbsdeAhat|^2))\,
    \mathbb{E}\left[ \sup_{t\in\mfT} \left|\fbsdeX_t - \tilde\fbsdeX_t\right|^2 \right] 
    \\
    & \qquad\quad
    +
    T^2\,(2|\fbsdeB|^2
    +30|\fbsdeBhat|^2)\,
    \mathbb{E}\left[ \sup_{t\in\mfT} \left|\fbsdeY_t - \tilde\fbsdeY_t\right|^2 \right]
\end{align}
\end{subequations}

Given  the conditions on $|\fbsdeG|$ and $T$ it follows that  $\Gamma$ is a contraction mapping and, by the Banach fixed point theorem in the Banach space $(\spaceK,||\cdot||_{\spaceK})$, we have that there exists a unique fixed point. 
It is then easy to see that the fixed point is the solution to the FBSDE system in \eqref{eqn:fbsde-system}.
\end{proof}

The norms $|{\fbsdeA}|, |{\fbsdeB}|,|{\fbsdeAhat}|, |{\fbsdeBhat}|$ are easy to compute for specific  examples;\footnote{
The expressions for $|{\fbsdeA}|, |{\fbsdeB}|$, and $ |{\fbsdeBhat}|$ are easy to obtain in general, but  $|{\fbsdeAhat}|$ is more involved. For example, for $|{\fbsdeA}|, |{\fbsdeB}|$, and $ |{\fbsdeBhat}|$, the general expressions are  $|\fbsdeA| =  \decayB$, $|\fbsdeB| = \sqrt{3 + \instantB^2 + \sqrt{5 - 2\,\instantB^2 + \instantB^4}}/\sqrt{2}$, and 
\begin{align}
    |\fbsdeBhat| &= \max\Bigg\{\frac{\instantB}{2\,\tempI}, \sqrt{
4\,\tempB^2\, \decayB^2 + 
 \instantB^2 (1 + \decayB^4) - \sqrt{-16 \,\tempB^2\,\instantB^2\,\decayB^2 + (4\,\tempB^2\,\decayB^2 + 
    \instantB^2\,(1 + \decayB^4))^2}}/(2 \sqrt{2}\, \tempB), \\
    &\qquad\qquad \sqrt{
4\,\tempB^2\,\decayB^2 + 
 \instantB^2\, (1 + \decayB^4) + \sqrt{-16 \,\tempB^2\, \instantB^2\,\decayB^2 + (4\,\tempB^2\,\decayB^2 + 
    \instantB^2\,(1 + \decayB^4))^2}}/(2\, \sqrt{2}\, \tempB)\Bigg\}\,.
\end{align}
}
for example, in the numerical experiments below, $|\fbsdeA|=0$, $|\fbsdeB|=1.618$, $|\fbsdeAhat|= 1$, and $|\fbsdeBhat|= 0.5$.

\section{Closed-form solution to the Nash equilibrium}\label{sec: closed form Nash vector not}
Let $\ell:\mfT\times \Omega\to \R^3$ and $\riccatiP:\mfT\to\R^{3\times 3}$  be given by
\begin{align}
{\ell}_t = 
\begin{pmatrix}
    \g_t\\
    \h_t\\
    \f_t
\end{pmatrix},
\qquad
\riccatiP_t = 
\begin{pmatrix}
    \gB_t & \gI_t & \gY_t\\
    \hB_t & \hI_t & \hY_t\\
    \fB_t & \fI_t & \fY_t
\end{pmatrix},
\end{align}
and consider the ansatz 
\begin{equation}
\fbsdeY_t = \ell_t + \riccatiP_t\,\fbsdeX_t\,.
\end{equation}
It follows that 
\begin{align}
\diff \fbsdeY_t 
&= \diff \ell_t +  \riccatiP'_t\,\fbsdeX_t\,\diff t + \riccatiP_t\,\diff \fbsdeX_t\\
&= \diff \ell_t +  \riccatiP'_t\,\fbsdeX_t\,\diff t + \riccatiP_t\,\left( \fbsdeA\,\fbsdeX_t + \fbsdeB\,\fbsdeY_t  +\fbsdeb_t\right)\,\diff t\\
&=\diff \ell_t +  \riccatiP'_t\,\fbsdeX_t\,\diff t + \riccatiP_t\,\left( \fbsdeA\,\fbsdeX_t + \fbsdeB\,\left(\ell_t + \riccatiP_t\,\fbsdeX_t\right)  +\fbsdeb_t\right)\,\diff t\,,
\end{align}
which together with \eqref{bsde of Y} implies that  
\begin{equation}
\diff \ell_t +  \riccatiP'_t\,\fbsdeX_t\,\diff t + \riccatiP_t\,\left( \fbsdeA\,\fbsdeX_t + \fbsdeB\,\left(\ell_t + \riccatiP_t\,\fbsdeX_t\right)  +\fbsdeb_t\right)\,\diff t = \left( \fbsdeAhat\,\fbsdeX_t + \fbsdeBhat\,\left(\ell_t + \riccatiP_t\,\fbsdeX_t\right) + \fbsdebhat_t \right)\,\diff t + \fbsdesigma_t\,\diff \fbsdeM_t\,.
\end{equation}
Rearranging the above expression and collecting terms with $\fbsdeX$ we obtain that 
\begin{equation}
\begin{split}
    0&=\left\{\diff \ell_t +  \left[\riccatiP_t\left(\fbsdeB\,\ell_t + \fbsdeb_t\right) - \fbsdeBhat\,\ell_t - \fbsdebhat_t\right]\diff t  - \fbsdesigma_t\,\diff \fbsdeM_t\right\}
    \\
&\qquad + \left\{ \riccatiP'_t  + \riccatiP_t\,\fbsdeA +  \riccatiP_t\,\fbsdeB\,\riccatiP_t - \fbsdeAhat - \fbsdeBhat\,\riccatiP_t \right\}\fbsdeX_t\,\diff t\,.
\end{split}
\end{equation}
Thus, showing existence and uniqueness of the solution to the FBSDE in \eqref{eqn:fbsde-system} boils down to showing existence and uniqueness of the   non-Hermitian matrix Riccati differential equation
\begin{subequations}
\begin{align}
0 &=\riccatiP'_t  + \riccatiP_t\,\fbsdeA +  \riccatiP_t\,\fbsdeB\,\riccatiP_t - \fbsdeAhat - \fbsdeBhat\,\riccatiP_t \,,\\
\riccatiP_T &= \fbsdeG\,,
\end{align}
 \label{eq:riccati new}
\end{subequations}    
and the linear BSDE 
\begin{subequations}
\begin{align}
0 &= \diff \ell_t +  \left[\riccatiP_t\left(\fbsdeB\,\ell_t + \fbsdeb_t\right) - \fbsdeBhat\,\ell_t - \fbsdebhat_t\right]\diff t  - \fbsdesigma_t\,\diff \fbsdeM_t\,,\\
\ell_T &= 0\,.
\end{align}\label{eq:bsde new}
\end{subequations}

Subsection \ref{sec: decay=0 and runnB=0} shows existence and uniqueness of a solution to \eqref{eq:riccati new} when $\decayB=\runnB=0$ under no further restrictions of model parameters, and Subsection \ref{sec: decay=0} shows existence and uniqueness when $\decayB=0$ under four additional mild restrictions on model parameters. Lastly, as mentioned before,  when $\decayB>0$ existence and uniqueness of the FBSDE follow from the assumption that $T$ is small enough; see Theorem \ref{thm: existence uniqueness} for the bound on the time horizon $T$.

\subsection{No resilience and no running penalty for the broker}\label{sec: decay=0 and runnB=0}
Next,  when the values of model parameters $\decayB$ and $\runnB$ are zero, we show existence and uniqueness of a solution to \eqref{eq:riccati new} with no additional constraints on the values of the remaining model parameters. Let
\begin{equation}
    {C}=\begin{pmatrix}
        0 & 0 & 0 \\
        0 & 0 & 0 \\
        0 & 0 & 1
    \end{pmatrix}\,,\qquad    {D}=\begin{pmatrix}
        -1 & 0 & 0 \\
        0 & -1 & 0 \\
        0 & 0 & 0
    \end{pmatrix}\,,
\end{equation}
and use Theorem 2.3 in \cite{freiling2000non} with the matrices ${C}$ and ${D}$ above. 
More precisely, using the notation of \cite{freiling2000non}, we have that $B_{11} = \fbsdeA$, $B_{12} = \fbsdeB$, $B_{21}= \fbsdeAhat$, and $B_{22} = \fbsdeBhat$. Then, it follows that for our choice of $C$ and $D$, 
\begin{equation}
    {C}+{D}\,\fbsdeG + \fbsdeG^\intercal\,{D}^\intercal = \begin{pmatrix}
       \frac{2\,\varphi}{\tempB} & 0 & 0\\
       0 &  \frac{2\,\termpI}{\tempI} & 0\\
       0 & 0 & 1
    \end{pmatrix}  >  0\,,
\end{equation}
which is positive definite because  $\varphi,\termpI > 0$. Next, matrix ${L}$ in Theorem 2.3 of \cite{freiling2000non} is given by 
\begin{equation}
    {L}=   
\begin{pmatrix}
{C}\,B_{11} + {D}\,B_{21} & {C}\,B_{12} +B_{11}^\intercal {D} + {D}\,B_{22}\\
  0 & B_{12}^\intercal {D}
\end{pmatrix} 
\,,
\end{equation}
and a short calculation (using Sylvester's criterion) shows that ${L}+{L}^\intercal\leq 0$. Similar to part II of the proof of Theorem 3.5 in \cite{casgrain2020mean}, the solution is bounded because the solution exists and is continuous in $[0,T]$; therefore,  from Theorem 3.1 in \cite{freiling2002survey},  the unique solution takes the form
\begin{align}\label{eq: Riccati solution}
    \riccatiP_t =  \mathcal{T}_t\,{R}^{-1}_t\,,
\end{align}
where the pair ${R}_t,\mathcal{T}_t$ is the unique solution to the linear system of differential equations 
\begin{align}
    \frac{\diff }{\diff t}
    \begin{pmatrix}
        {R}_t\\
        \mathcal{T}_t
    \end{pmatrix}
    = 
    \begin{pmatrix}
        B_{11} & B_{12}\\
        B_{21} & B_{22}
    \end{pmatrix}
    \,
    \begin{pmatrix}
        {R}_t\\
        \mathcal{T}_t
    \end{pmatrix}\,,\qquad \begin{pmatrix}
        {R}_T\\
        \mathcal{T}_T
    \end{pmatrix} = \begin{pmatrix}
        I\\
        \fbsdeG
    \end{pmatrix}\,.
\end{align}

\subsection{No resilience}\label{sec: decay=0}
In this subsection we study the case $\decayB=0$ in more detail. In particular, when $\decayB=0$, we establish a one-to-one correspondence between four of the functions in  the Nash equilibrium we find, and four functions in the mean-field Nash equilibrium in \cite{bergault2024mean}. When
 $\decayB=0$,  the system of equations for $\f,\g,\h,\fB,\gB,\hB,\fI,\gI,\hI,\fY,\gY,\hY$ simplifies, and we have that $\gY=\hY=\fY=0$. 
As a consequence, the equations for $\gB,\hB,\gI,\hI$ decouple from the rest. 
After a renaming of terms, we obtain the system of ODEs as that characterised in equation (3.10) in \cite{bergault2024mean}. 
That is, after renaming equations and parameters, the solution to $\gB,\hB,\gI,\hI$ above, is  equivalent to that of $g^b,g^c,h^b,h^c$ in their paper. Existence and uniqueness of the solution follow from their results under the mild condition $\instantB \leq 2\,\tempI,\,2\,\tempB,\,4\,\runnI,\,4\,\runnB$.  The remainder of their results are different to the ones in this paper, this is because in their model each player observes two signals, a private signal and a common signal, thus, in their model,  the optimal strategy of any given informed trader is different from the one we derive here by virtue of the formulation.

\subsection{The solution to the BSDE}
Next, we use the solution to the non-Hermitian matrix Riccati differential equation to find a solution to the BSDE \eqref{eq:bsde new} that $\ell$ must satisfy. Write \eqref{eq:bsde new} as
\begin{align}
\diff {\ell}_t = \left({a}_t +  \Lambda_t\,{\ell}_t\right)\diff t + {\sigma}_t\,\diff  \fbsdeM_t\,,
\end{align}
where $a_t = \fbsdebhat_t - \riccatiP_t\,\fbsdeb_t$ and $\Lambda_t = \fbsdeBhat - \riccatiP_t\,\fbsdeB$, and let ${\zeta}_t$ be the unique solution to
\begin{align}
\diff {\zeta}_t = -{\zeta}_t\,{\Lambda}_t\,\diff t\,,
\end{align}
with initial condition ${\zeta}_0=I\in\R^{2\times 2}$.  Then, 
\begin{align}
\diff \left({\zeta}_t\,{\ell}_t \right) & = {\ell}_t \,\diff {\zeta}_t + {\zeta}_t\,\diff {\ell}_t\\
& = - {\ell}_t\, {\zeta}_t\,{\Lambda}_t\,\diff t + {\zeta}_t\,\left({a}_t +  \Lambda_t\,{\ell}_t\right)\diff t - {\zeta}_t\, {\sigma}_t\,\diff  \fbsdeM_t\\
& ={\zeta}_t\,{a}_t\,\diff t - {\zeta}_t\,{\sigma}_t\, \diff  \fbsdeM_t\,,
\end{align}
which implies that
\begin{align}
 - {\zeta}_t\, {\ell}_t
& =\int_t^T {\zeta}_u\, {a}_u\,\diff u + \int_t^T {\zeta}_u\, {\sigma}_u\, \diff  \fbsdeM_u\,,
\end{align}
because ${\zeta}_T\,{\ell}_T = 0$. Thus, by taking conditional expectations, the solution to ${\ell}_t$ is
\begin{align}\label{eq: ell cond expectation}
{\ell}_t
& = - \E\left[\int_t^T {\zeta}^{-1}_t\, {\zeta}_u\, {a}_u\,\diff u\,\big|\,\mcF_t\right] \,.
\end{align}
We conclude that the Nash equilibrium $\nuBstar,\nuIstar$, which are the first two components of $\fbsdeY$, is given by 
\begin{align}\label{eq: closed form for strategies Nash compact}
   \fbsdeY_t = {\ell}_t + \riccatiP_t\,\fbsdeX_t\,,
\end{align}
where $\fbsdeX$ is the unique solution to the forward equation \eqref{eqn: sde-for-X}.

\section{Numerical results}\label{sec: numerics}
In this section, we study the optimal trading strategies and the Nash equilibrium between the broker and the informed trader. We discretise the trading window $[0,T]$, with $T=1$,   into  10,000 steps and perform one million simulations. The terminal penalty  parameters are $\termpI = \termpB = 1$, the instantaneous price  impact parameters are $\tempB = 1.2\times 10^{-3}$ and $\tempI = \tempU = 1\times 10^{-3}$, the running penalties are $\runnI = \runnB = 0$, and the transient price impact parameters are $\instantB = 10^{-3}$, and $\decayB = 0$. The midprice starts at $S_0= 100$ and  $ M^S_t=\sigma^S\, W^S_t$, where $\sigma^S = 1$ and $W^S_t$ is a Brownian motion. The signal $\alpha_t$ follows the OU process $\diff \alpha_t = -\kappa^\alpha \,\alpha_t\,\diff t + \sigma^\alpha \,\diff W^\alpha_t $, where $\kappa^\alpha = 5$, $\sigma^\alpha = 1$, $\alpha_0 = 0$, and $W^\alpha_t$ a Brownian motion. Finally, the trading speed of the uninformed trader follows the OU process $\diff \nuU_t = - \kappa^\nuU \,\nuU_t\,\diff t + \sigma^\nuU\,\diff W^\nuU_t$, where $\kappa^\nuU = 15$, $\sigma^\nuU = 100$,
$\nuU_0 = 0$, and $W^\nuU$ is a Brownian motion. The three Brownian motions $W^\nuU$, $W^S$, and $W^\alpha$ are independent, for simplicity.  {\new{The above parameters allow us to use the existence and uniqueness result we obtained at the level of the matrix Riccati differential equation  in Section \ref{sec: decay=0 and runnB=0}.}}
Next, we provide a closed-form representation of \eqref{eq: ell cond expectation}, given by \begin{equation}
    {\ell}_t = \begin{pmatrix}
				\g_t\\
				\h_t\\
                \f_t
			\end{pmatrix} = 
\begin{pmatrix}
    g^1(t)\,\alpha_t + g^2(t)\,\nuU_t\\
    h^1(t)\,\alpha_t + h^2(t)\,\nuU_t \\
    f^1(t)\,\alpha_t + f^2(t)\,\nuU_t \\
\end{pmatrix}\,,
\end{equation}
where $g^{1,2}$, $h^{1,2}$, and $f^{1,2}$ are deterministic functions characterised by the unique solution to the linear system of ODEs 
\begin{subequations}
    \begin{align}
0 &= \frac{\diff g^1_t}{\diff t} + \gB_t\,\left(g^1_t -h^1_t\right) + \gI_t\,h^1_t + g^1_t\,\gY_t\,\instantB - g^1_t\, \kappa^\alpha + \frac{1}{
 2\, \tempB} + \frac{h^1_t\,\instantB}{2\, \tempB}\,,\\
0 &= \frac{\diff g^2_t}{\diff t} - \gB_t + \gB_t\,\left(g^2_t  - h^2_t\right) + \gI_t\, h^2_t + g^2_t \,\gY_t\, \instantB - g^2_t\, \kappa^\nuU + \frac{\instantB}{2\, \tempB} + \frac{h^2_t\, \instantB}{2\, \tempB} \,,\\
0 &= \frac{\diff h^1_t}{\diff t}  + \hB_t\,\left(g^1_t - h^1_t\right) + h^1_t\, \hI_t + g^1_t \,\hY_t\, \instantB - h^1_t\, \kappa^\alpha + \frac{1}{
 2\, \tempI} + \frac{g^1_t \,\instantB}{2\,\tempI}\,,\\
0 &= \frac{\diff h^2_t}{\diff t} -\hB_t + \hB_t\left(g^2_t  - h^2_t\right)  + h^2_t \,\hI_t + g^2_t\, \hY_t \,\instantB - h^2_t\, \kappa^\nuU +\frac{ 
 g^2_t \,\instantB}{2\, \tempI}\,, \\
0 &= \frac{\diff f^1_t}{\diff t} + \fB_t\left( g^1_t - h^1_t\right) + \fI_t\, h^1_t + \fY_t g^1_t \,\instantB - f^1_t\, \kappa^\alpha\,, \\
0 &= \frac{\diff f^2_t}{\diff t} -\fB_t + \fB_t \left(g^2_t - h^2_t\right) + \fI_t\, h^2_t + \fY_t\, g^2_t \,\instantB - f^2_t \,\kappa^\nuU\,,
\end{align}%
\end{subequations}%
with terminal conditions $g^{1,2}_T = 0$, $h^{1,2}_T = 0$, and $f^{1,2}_T = 0$.
Given that $\decayB=0$, we require only $\gB,\gI,\hB,\hI:[0,T]\to\R$ as then, $\fB,\fI,\fY,\gY,\hY$ do not enter the optimal strategies of the broker and informed traders.

Figure \ref{fig: Sample paths} shows a sample path for the trading signal $\alpha$, the optimal trading speeds of the informed trader ($\nuIstar$) and the broker ($\nuBstar$) together with their inventories $\QI[I,\nuIstar]$ and $\QB[B,\nuBstar]$. Additionally, we include the $5\%$ and $95\%$ quantiles through time for these quantities.

\begin{figure}[H]
\centering
\includegraphics[width=0.32\textwidth]{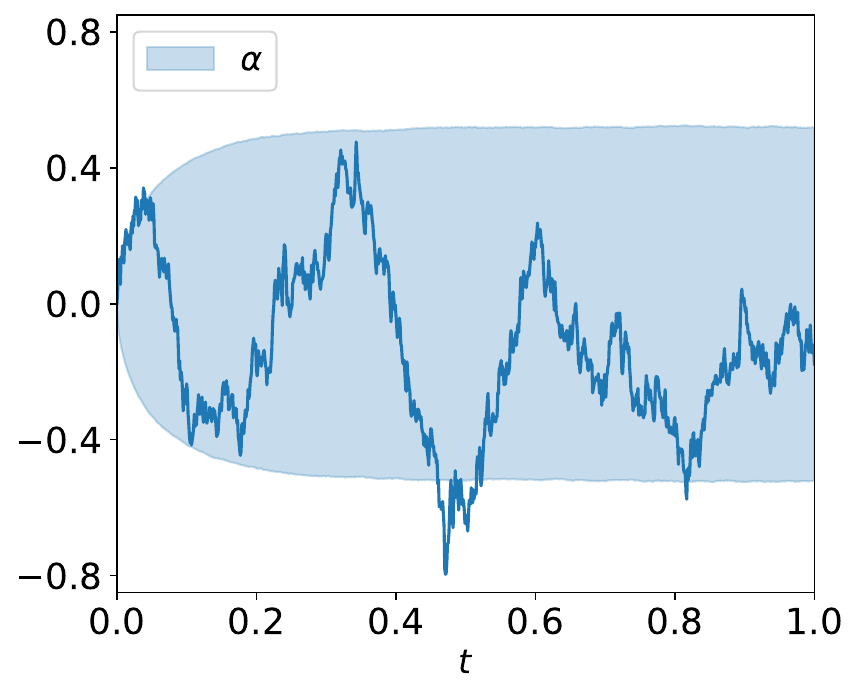} 
\hfill
\includegraphics[width=0.32\textwidth]{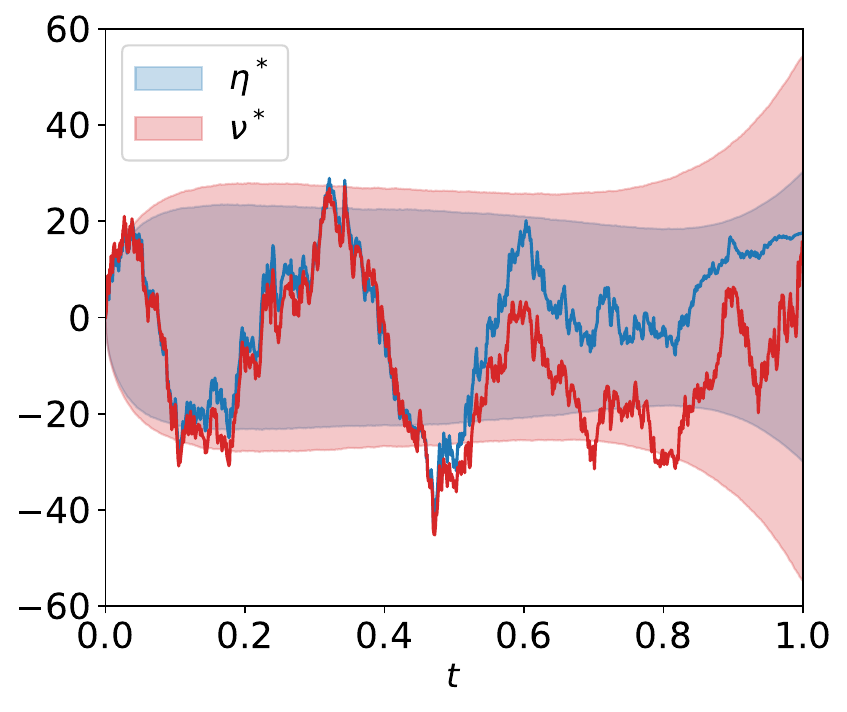}
\hfill
\includegraphics[width=0.32\textwidth]{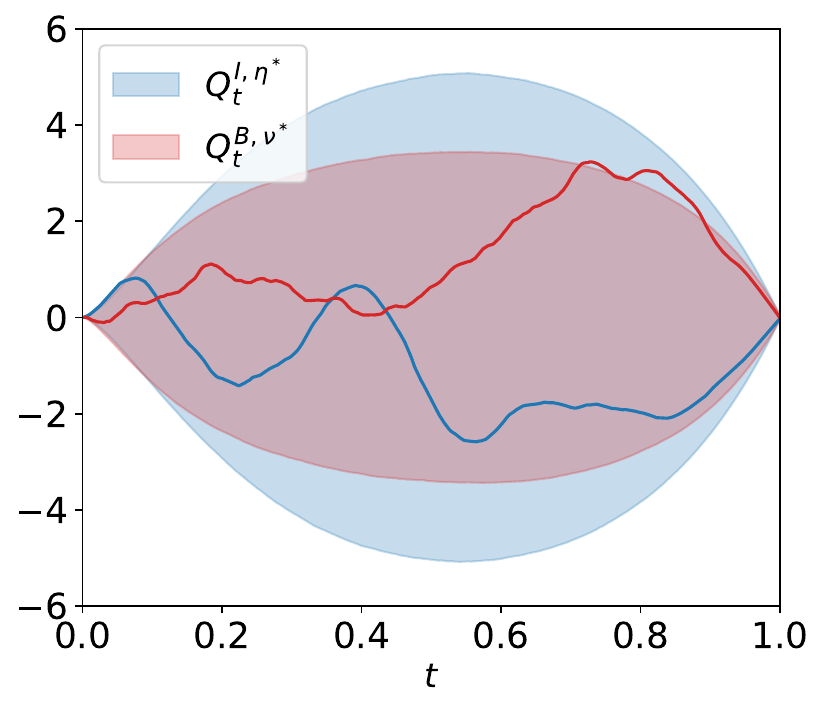}
\caption{Selected sample paths for the main processes that play a role in the Nash equilibrium.}
\label{fig: Sample paths}    
\end{figure}

{\new{Both players end the trading day without exposure on their inventories (due to the terminal penalties on inventories in the performance criteria of the players).Also, the path of the broker's inventory is rougher     than that of the informed trader's. This is due to the trading of the uninformed trader that only (directly) affects the inventory of the broker. Lastly, the trading speeds of the broker and the  informed trader 
align with the signal, which verifies the assumptions on the information sets of both players. In Figure \ref{fig: uninformed zero}, we set the uninformed trading speed to zero, i.e., $\sigma^\xi=0$. The figure confirms that the broker's inventory path becomes smoother, and the trading speeds of the informed trader and the broker align even further. Hence, the broker effectively offloads the inventory that the informed traders passes to them on to the market, however, the broker does keep a small position to benefit from the trading signal.}}
\begin{figure}[H]
\centering
\includegraphics[width=0.32\textwidth]{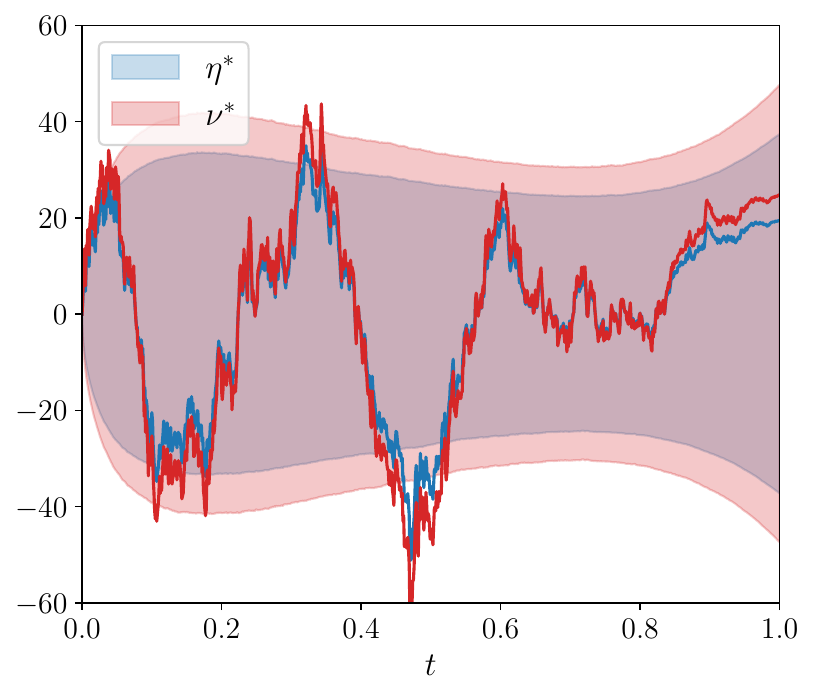}
\includegraphics[width=0.32\textwidth]{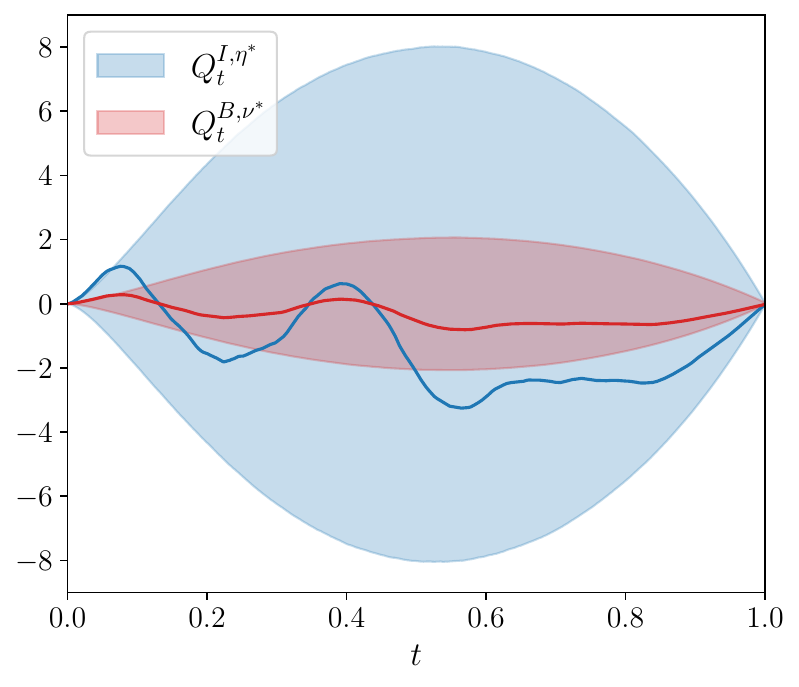}
\caption{Sample paths of the trading speeds and the inventories of the informed trader and the broker for when the uninformed trader does not trade.}
\label{fig: uninformed zero}    
\end{figure}

\subsection{Transient price impact}

{\new{
Next, we study the effect that the transient impact parameters $\decayB$ and $\instantB$ have in the Nash equilibrium.
More precisely, we stress the values of model parameters $\decayB$ and $\instantB$. The remaining values of model parameters are those reported at the beginning of Section \ref{sec: numerics}. 
We study the cases where $\decayB\in\{0,1,2,4\}$ and $\instantB\in\{0.001,0.1\}$, shown in Figure \ref{fig:DecayInstant}. Unfortunately, when $\decayB>0$, one cannot employ the existence and uniqueness results from Section \ref{sec: decay=0 and runnB=0}, and the results from Theorem \ref{thm: existence uniqueness} impose restrictive bounds on $T$ as we increase the value of $\decayB$. For example, when $\decayB=1,\instantB\in\{0.001,0.1\}$, the bound on $T$ is reasonably set at around 0.2, whereas when $\decayB=4,\instantB\in\{0.001,0.1\}$, one needs to take  $T<  7\times 10^{-5}$ so that existence and uniqueness follows from Theorem \ref{thm: existence uniqueness}. Although we do not have a theorem to guarantee existence and uniqueness for these parameters, none of the numerical approximations explode. In what follows we carry out the analysis with these values so that the effects of these parameters is apparent; we leave sharper bounds on the short time horizon for future work.
\begin{figure}[H]
    \centering
    \includegraphics[width=0.3\linewidth]{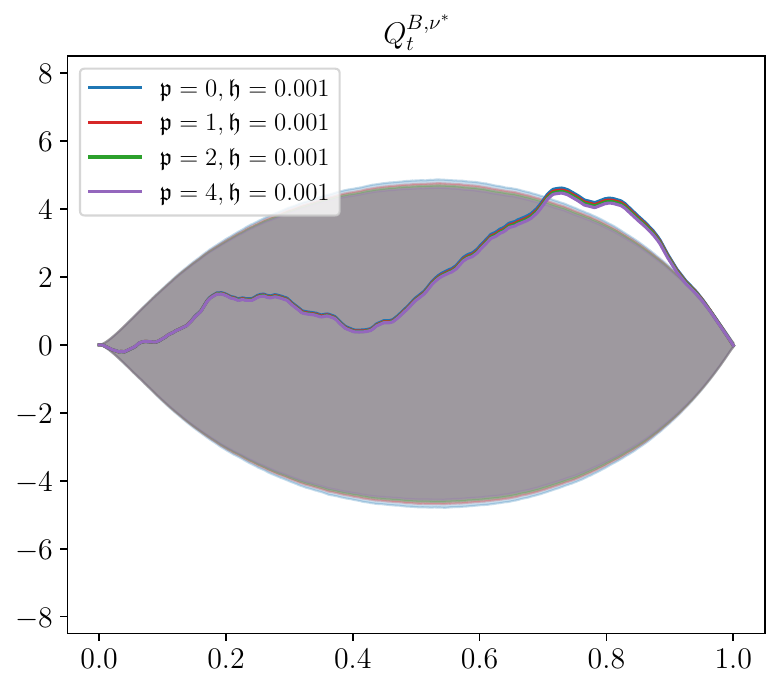}
    \includegraphics[width=0.3\linewidth]{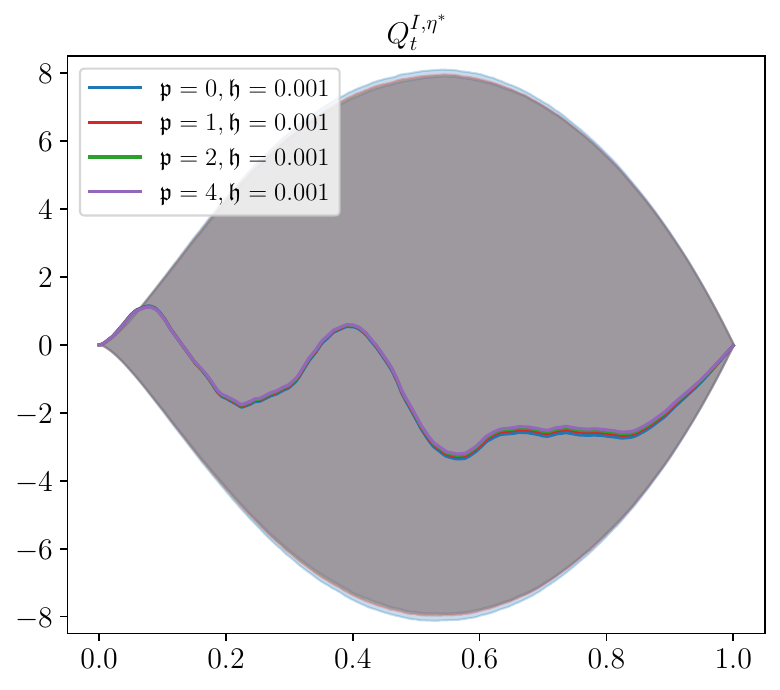}\\
    \includegraphics[width=0.3\linewidth]{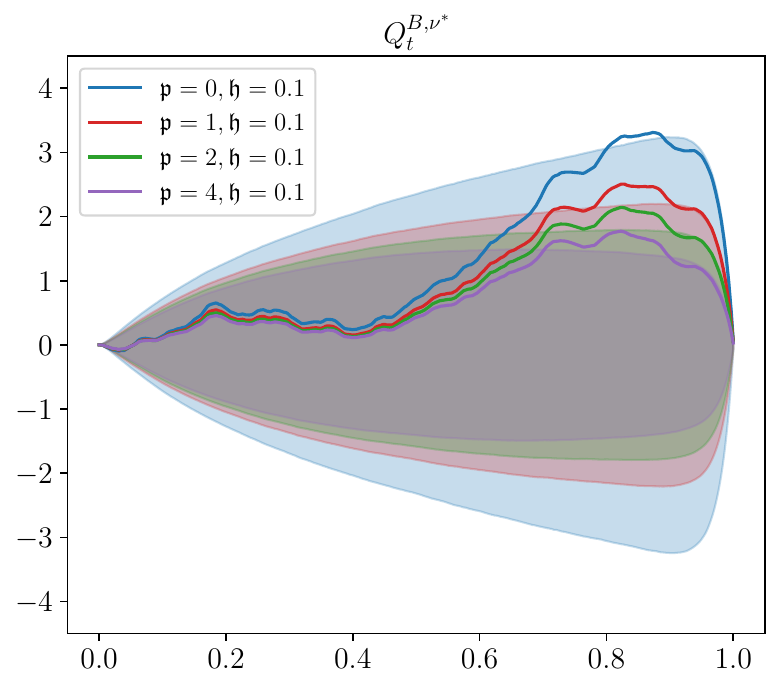}
    \includegraphics[width=0.3\linewidth]{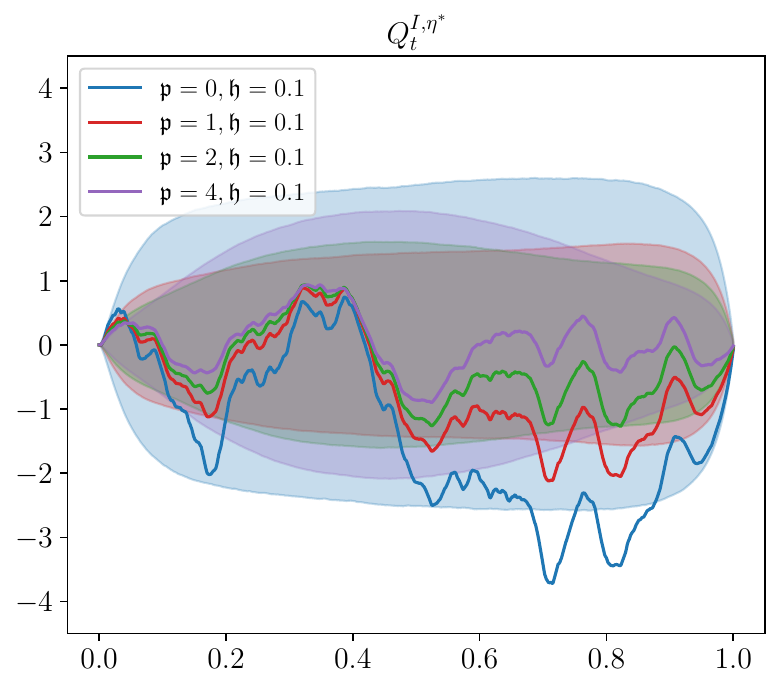}
    \caption{Trajectory of the inventories of the broker (left panel) and the informed trader (right panel). The top panels are for $\instantB=0.001$ and the bottom panels are for $\instantB=0.1$. The plots show the cases for $\decayB\in\{0,1,2,4\}$. }
    \label{fig:DecayInstant}
\end{figure}

The informed trader's strategy consists of two key drivers: direction of the signal and price impact. In equilibrium, there is competition between price discovery and the `riding' of price impact.  There are situations where these two drivers may go in opposite directions. For example, when the price impact is strong and persistent, and the broker's strategy is to trade in the direction of the market impact, the   informed trader may find it optimal to trade in the opposite direction of the signal and trade in tandem with the broker.

The bottom right panel of Figure \ref{fig:DecayInstant} shows a fundamental change  in the trajectories of the inventory of the informed trader depending on the value of the decay parameter $\decayB$. When the value of $\decayB$ is small (see e.g., $\decayB=0$) the informed trader deploys what looks like a ``pump-and-dump'' strategy. That is, the informed trader exaggerates the trading on the signal to profit from the price impact. On the other hand, when the value of $\decayB$ is large (see e.g., $\decayB=4$) the shape of the inventory trajectory suggests that the informed trader follows  the signal and does not ride the price impact. Thus, there is a range of values for which it is more beneficial to ride the price impact.

}}

\subsection{From incomplete information and model misspecification to perfect information}

Here, we compare the two-stage optimisation  in \cite{cartea2022broker}  with the Nash equilibrium we find. {\new{We employ the same model parameters as those above and set the ambiguity aversion parameters in their paper to $10^{-7}$ so the effect is negligible}}.\footnote{{\new{The results are similar for  smaller values of the ambiguity aversion.}}} We obtain 100,000 sample paths for the mark-to-market of the  broker and informed trader when following the strategies from the Nash equilibrium and those from the two-stage optimisation, all using the same set of Brownian paths. Figure \ref{fig: contours} shows the resulting kernel density estimator for all six pairwise quantities, where Nash-informed (Nash-broker) refers to the mark-to-market of the informed (broker) trader when following the Nash-equilibrium strategy, and $\Delta$ informed ($\Delta$ broker) refers to the the difference between the mark-to-market when following the Nash equilibrium and
the mark-to-market when following the two-stage optimisation strategy for the informed trader (broker).

\begin{figure}[H]
\begin{center}
\includegraphics[width=0.5\textwidth]{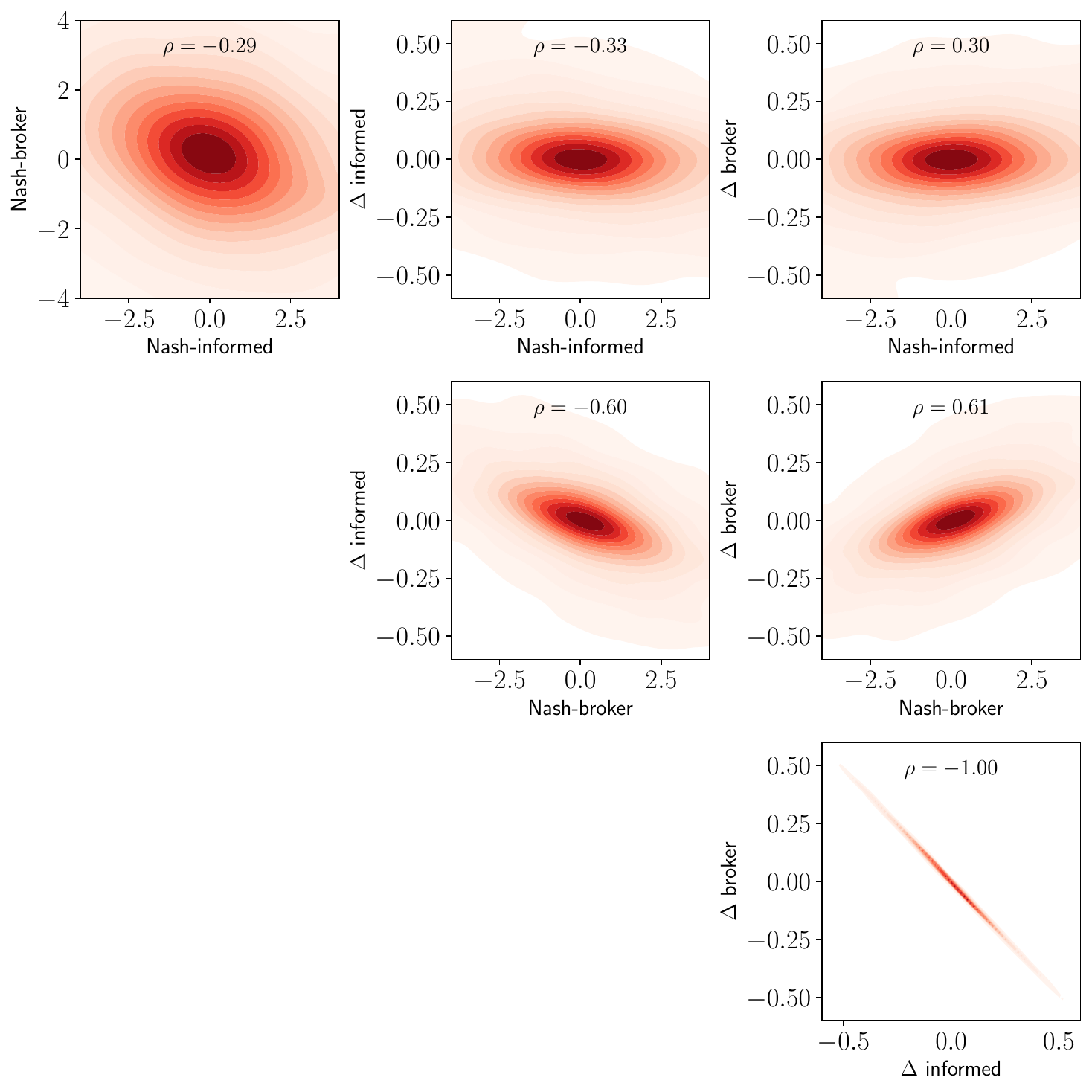}
\caption{Contour plots of the kernel density estimator for the Nash equilibrium found here versus the two-stage optimisation in \cite{cartea2022broker}. We report the correlation $\rho$ between the two variables.}
\label{fig: contours}
\end{center}
\end{figure}

{\new{The mean mark-to-market (MtM) of the informed trader increases by roughly 37 basis points (bps) when going from the two-stage optimisation in \cite{cartea2022broker} to the Nash equilibrium, whereas the broker decreases their mean MtM by 91 bps. A $t$-test shows, with 90\% confidence, that the average MtM of the samples generated by both the Nash and two-stage optimisation are statistically different (for both the informed trader and the broker).\footnote{{\new{The $p$-values are 0.07 and 0.03 for the informed and the broker tests respectively. }}}  
This difference provides an upper bound for how profitability of the players' strategies change when information becomes public. 
Thus, in our model (for the parameters above), the maximum improvement that the informed trader can obtain is 37 bps  which is largely paid by the broker; note that this is four orders of magnitude higher than transaction costs in foreign exchange markets.}}\footnote{Transaction costs in spot foreign exchange are roughly \$3 per \$1 million traded; see \cite{cartea2021optimal}.} {\new{This result states that when going from the strategies in  \cite{cartea2022broker} to the perfect information Nash equilibrium,  the informed trader gains more from knowing about the activity of the broker (and their impact in prices) than what the broker gains from knowing the signal. Indeed, in  \cite{cartea2022broker} the broker benefits from knowing the two groups of traders (informed and uninformed), and the optimal strategy externalises a fair deal of the informed trader's order flow. The added externalisation that the broker carries out when knowing the signal is less beneficial than what the informed trader obtains from knowing the information set at the disposal of the broker, in particular, her impact on prices, and her inventory (which largely determines the future offloading that will take place in the market). This provides insights into which player has more incentives to hide their order flow.

Our results also provide insights into what happens in maker-taker relationships in  markets where the market maker does not benefit from anonymity.
Such perfect information trading environments are a reality in the decentralised finance space --- in particular in  automated market maker trading venues. More precisely, in these trading venues, where billions of dollars are traded every day,\footnote{On 26 November 2024, the volume traded in Uniswap V3 was slightly over \$4 billion USD.\\ See \url{https://defillama.com/dexs/uniswap-v3} for the latest figures.} the trading activity of both liquidity providers and liquidity takers, is visible to everyone. Within our setup, the broker would be the liquidity providers in the liquidity pool and the external venue could be Binance. 

}}

\section{Acknowledgements}
SJ would like to thank the National Sciences and Engineering Research council of Canada for partially supporting this research through grants RGPIN-2018-05705 and RGPIN-2024-04317.

\bibliographystyle{apalike}
\bibliography{References}

\end{document}